\documentclass[final,leqno,onetabnum]{siamltex704}

\usepackage{amsmath}
\usepackage{amsfonts}
\usepackage{amssymb}
\usepackage{epsfig}
\usepackage{latexsym}
\usepackage{exscale}
\usepackage{subfigure}
\usepackage[algo2e,linesnumbered,vlined,ruled]{algorithm2e}

\usepackage{graphicx}
\usepackage[active]{srcltx}

\usepackage{color}
\usepackage[normalem]{ulem}

\def\k{\kern .5em}
\def\er{\kern .2em}

\newcommand{\comm}[1]{{\color{red}#1}}

\newcommand{\be}{\begin{equation}}
\newcommand{\ee}{\end{equation}}
\newcommand{\ba}{\begin{array}}
\newcommand{\ea}{\end{array}}
\newcommand{\bea}{\begin{eqnarray*}}
\newcommand{\eea}{\end{eqnarray*}}
\newcommand{\bean}{\begin{eqnarray}}
\newcommand{\eean}{\end{eqnarray}}

\newcommand{\ddt}{{\frac{\mathrm{d}}{\mathrm{d}t}}}

\newcommand{\Rnp}{\mathbb{R}^{n \times p}}
\newcommand{\Mnp}{\mathcal{M}_n^p}

\newcommand{\T}{\mathcal{T}}
\newcommand{\cR}{\mathcal{R}}

\newcommand{\st}{\mbox{s.t.}}
\newcommand{\tr}{\mathrm{tr}}
\newcommand{\QR}{\mathrm{QR}}
\newcommand{\WY}{\mathrm{WY}}
\newcommand{\iprod}[2]{\left \langle #1, #2 \right \rangle }

\newtheorem{Theorem}{Theorem}[section]

\newtheorem{Proposition}{Proposition}[section]

\newtheorem{Assumption}{Assumption}[section]

\newcommand{\lc}{\mathrel{\raise2pt\hbox{${\mathop<\limits_{\raise1pt\hbox{\mbox{$\sim$}}}}$}}}
\newcommand{\gc}{\mathrel{\raise2pt\hbox{${\mathop>\limits_{\raise1pt\hbox{\mbox{$\sim$}}}}$}}}
\newcommand{\ec}{\mathrel{\raise1pt\hbox{${\mathop=\limits_{\raise2pt\hbox{\mbox{$\sim$}}}}$}}}

\renewcommand{\theequation}{\arabic{section}.\arabic{equation}}

\title{Gradient type optimization methods for electronic structure calculations
\thanks{The work of X. Zhang, J. Zhu and A. Zhou was supported by the Funds for Creative Research Groups of China under
grant 11021101, the National Basic Research Program of China under
grant 2011CB309703, and the National Center for Mathematics and
Interdisciplinary Sciences, Chinese Academy of Sciences. The work of Z. Wen was
supported in part by NSFC grant 11101274.}}

\author{Xin Zhang\thanks{LSEC,
Institute of Computational Mathematics and Scientific/Engineering
Computing, Academy of Mathematics and Systems Science, Chinese
Academy of Sciences, Beijing 100190, China (xzhang,jwzhu,azhou@lsec.cc.ac.cn).}
\and Jinwei Zhu$^\dagger$
\and Zaiwen Wen\thanks{Beijing International Center for Mathematical Research,
Peking University, Beijing 100871, China. }
\and Aihui Zhou$^\dagger$
}

\begin{document}

\maketitle

\begin{abstract}
The density functional theory (DFT) in electronic structure calculations can be
formulated as either a nonlinear eigenvalue or direct minimization problem. The
most widely used approach for solving the former  is the so-called
self-consistent field (SCF) iteration. A common observation is that  the
convergence of SCF is not clear theoretically while approaches with convergence
guarantee for solving the latter are often not competitive to SCF numerically.
In this paper, we study gradient type methods for solving the direct
minimization problem by constructing new iterations along the gradient on the
Stiefel manifold. Global convergence (i.e., convergence to a stationary point from any
 initial solution) as well as local convergence rate follows
from the standard theory for optimization on manifold directly.  A major computational
advantage is that the computation of linear eigenvalue problems is no longer needed. The main costs of our approaches arise from the assembling of the total energy functional and its gradient and the projection onto the manifold. These tasks are  cheaper than eigenvalue computation and they are often  more suitable for parallelization as
long as the evaluation of the total
energy functional and its gradient is efficient.
Numerical results show that they can
outperform SCF consistently on many practically large
systems. 
\end{abstract}


\begin{AM}
65N25, 65N30, 65N50, 90C30
\end{AM}

\pagestyle{myheadings}

\thispagestyle{plain}

\markboth {XIN ZHANG, JINWEI ZHU, ZAIWEN WEN AND AIHUI ZHOU} {OptESC}

\section{Introduction}\label{sec:intro}
\setcounter{equation}{0}
Electronic structure calculations have been widely used in chemistry,
materials science, drug design and nano-science over the past decades
because of its great advantage in predicting phase transformation
in various materials and a few other useful properties, such as optics, electric conductivity
and magnetics \cite{martin04}. The computational complexity of simulations using the
many-body Schr\"{o}dinger equation is extremely expensive. One of the most
fundamental advances is the Density Functional Theory (DFT) models. These
models can be divided into two classes: Orbital Free DFT
(OFDFT) model and Kohn-Sham DFT (KSDFT) model \cite{kohn-sham65, martin04, saad-chelikowshy-shontz-2010, wang-govind-carter99}. Both
of them can be formulated as either nonlinear eigenvalue  or direct
minimization problems under orthogonality constraints of wave functions,
where the former corresponds to the first-order
necessary optimality condition of the latter.

There are several approaches for discretizing the DFT  models. The plane
wave basis has been  popular  due to its advantage
in expressing the kinetic energy term and the Hartree potential in simple forms.
It has been used in a few softwares, such as ABINIT \cite{abinit} and VASP
\cite{vasp}. However, this type of basis is often suitable for periodic
systems which does not contain the realistic cases like vacancy and dislocation,
and it is usually not suitable for isolated systems, such as molecules and clusters.
Since the plane wave basis is globally defined in real space, it
may not be efficient for massively parallel computing \cite{motamarri2012, schauer-linder2013}.
The second type of approach is the so-called atom-centered basis set. Since
a few number of basis can provide satisfying results, it has been applied in
some other  softwares, such as Gaussian \cite{GAUSSIAN} and SIESTA \cite{sieata}.
On the other hand, it is not easy to construct a practically complete basis
\cite{losilla-sundholm2012}.
Another popular discretization is the real space
approaches, including the finite difference, finite element, finite volume,
and wavelet methods \cite{chen-gong-he-yang-zhou13, chen-he-zhou11, chen-gong-zhou10, dai-gong-yang-zhang-zhou11, fang-gao-zhou12,
gong-shen-zhang-zhou08, hirose-ono-fujimoto-tsukamoto05, motamarri2012, pask-sterne2005,
torsti2006, Tsuchida-tsukada95, zhang-shen-zhou-gong08, zhou07}.
They can handle
computational domains with complicate geometries and diversified
boundary conditions. In particular, the whole domain can be divided into low and
high resolution parts according to the frequency of the wave functions. 
A complete basis sets can always be chosen. Although their degrees of freedom are large, the
adaptive methods and multilevel methods 
\cite{martin04} can
be applied to reduce the computational complexity.

The most widely used method for solving the nonlinear eigenvalue formulation
is the so-called self-consistent field (SCF) iteration. Despite its popularity
there are two well known challenges in SCF. First, its computational cost is
dominated by the eigenvalue computation. For a system with  $N$ electrons,
the $N$ smallest eigenvalues and their associated eigenvectors of the Hamiltonian
must be computed at each iteration. Second, the convergence of SCF is
not guaranteed either theoretically or numerically. Its performance is
often unpredictable for large
scale systems with small band gaps, even with the assistance of the charge
density or potential mixing techniques \cite{johnson88,pulay80}. On the other
hand, optimization methods for direct minimization problems have been proposed
for electronic structure calculations \cite{
BendtZunger1982, Li-Nunes-Vanderbilt1993, millam-Scuseria1997, PayneTeterAriasJoannopoulos1992, Pfrommer1999, saad-chelikowshy-shontz-2010, VanHead2002}.
Trust region methods \cite{Francisco2004, FranciscoMartine2006,
Thogersen2004} substitute the linear eigenvalue problem in SCF
by the so-called trust-region subproblems, in which the objective function are
local quadratic approximations to the total energy functional.  Monotonic
reduction of the total energy can be achieved under a suitable update of the
trust-region radius. A direct constrained minimization (DCM) algorithm is
designed in \cite{yang-meza-wang07}, where the new search direction is built
from a subspace spanned by the current approximation to the optimal wave
function, the preconditioned gradient and the previous search direction.
Although many of these optimization based methods often have better theoretical convergence
 properties than SCF,  a common observation is that they are not competitive
 to SCF when the latter works and  they may converge slowly on large
 scale systems \cite{kresse-furtthmuller1996}.

In this paper, we study gradient type optimization methods
 derived from  the direct minimization formulation to solve both OFDFT and KSDFT models.
Essentially, our approaches construct new trial points along the gradient on the
Stiefel manifold, i.e., a manifold consisted of orthogonal matrices. The
orthogonality is preserved by   an operation called retraction
\cite{opt-manifold-book} on the Stiefel manifold. Consequently, all theoretical properties of
optimization on manifold  \cite{opt-manifold-book} can be applied to our approaches
 naturally.  A major computational advantage is that  eigenvalue
computation is no longer needed. The main costs of our approaches arise from the
assembling of the total energy functional and its gradient on manifold and the execution of retractions. These
tasks are   cheaper than eigenvalue computation and they are
often more suitable for parallelization. The numerical performance of our
gradient methods is further improved by the state-of-the-art
acceleration techniques such as Barzilai-Borwein steps and non-monotone
line search with global convergence guarantees.
  Our approaches can quickly reach the vicinity of an optimal solution and produce
a moderately accurate approximation, at least in our numerical examples. 

Our main contribution is the demonstration that the simple gradient type
methods can outperform SCF consistently on many practically meaningful systems
based on real space discretization.  Two types of retractions on Stiefel
manifold are investigated. One preserves  orthogonality by using a
Crank-Nicolson-like scheme proposed in \cite{zai-yin} for minimizing a
general differentiable function on Stiefel manifold.
The other adopts the QR factorization for orthogonalization explicitly with
respect to a gradient step on the Stiefel manifold. Due to
the orthogonality, every intermediate point along the gradient direction of the
Stiefel manifold has full rank which ensures the stability when the
Cholesky factorization is used to compute the QR factorization.
To the best of the authors' knowledge,  it is the first time that the QR-based
method is systematically studied for electronic
structure calculations.
Our methods are implemented within the software packages
Octopus \cite{octopus} for KSDFT and RealSPACES (developed by the State Key Laboratory of
Scientific and Engineering Computing of the Chinese Academy of Sciences)  for
OFDFT, respectively.
Numerical experiments show that they can be more efficient and robust than
SCF on many instances. High parallel scalability  can also be achieved.
 However, we should point out that their efficiency still depends on an
efficient evaluation of the total energy functional and their gradients.
Although the regularized SCF method in \cite{WenMilzarekUlbrichZhang2013}
can converge faster by using the
Hessian of the total energy functional,
it requires the gradient methods to compute the search direction from some
quadratic or cubic subproblems.


The rest of the paper is organized as follows. In Section
\ref{sec:prob}, we present the mathematical models of DFT models, their
associated discretization and the SCF iteration algorithm. Our gradient
type algorithms are described in Section~\ref{sec:alg}. Numerical results
are reported in Section~\ref{sec:num} to illustrate the efficiency of our
algorithms. Finally, some concluding remarks are given in Section
\ref{sec:conclusion}.

\section{Preliminaries} \label{sec:prob}
\subsection{The KSDFT model}
The KSDFT model was developed based on the theories established by Hohenberg, Kohn
and Sham \cite{kohn-sham65}. The essence of KSDFT
 is treating the many-body system as an equivalent system with
non-interacting electrons in an effective mean field governed by the
charge density \cite{motamarri2012}. Using this representation, KSDFT can
describe the ground state exactly in principle using single-electron
wave functions.

For a system with $M$ nuclei and $N$ electrons, let $r$ represent the
spatial coordinate in three dimensions. The charge density under the KSDFT
model is expressed as $\rho(r)=\sum_{i=1}^{N}|\phi_i(r)|^2,$
where $\phi_i(r)$ is the $i$-th wave function associated with the
non-interactive auxiliary system under the orthogonality constraint
$
\int_{\mathbb{R}^3}\phi_i(r)\phi_j(r)\mathrm{d}r=\delta_{ij},
$
where $\delta_{ij}$ is the Dirac delta function. The kinetic energy of the
system can be expressed by the $N$ independent orthonormalized wave
functions
\[
T_{KS}=\frac1{2}\sum_{i=1}^{N}\int_{\mathbb{R}^3}|
\nabla\phi_i(r)|^2\mathrm{d}r.
\]
The Hartree energy is written as
\[
E_H(\rho)=\frac{1}{2}\int_{\mathbb{R}^3}V_{H}(\rho)\rho(r)
\mathrm{d}r,
\]
where $V_{H}(\rho)=\displaystyle\int_{\mathbb{R}^3}\frac{\rho(r')}{|r-r'|}
\mathrm{d}r'$ is the classical electrostatic average interaction between
electrons.  The exchange-correlation energy is defined as
\[
E_{xc}(\rho)=\int_{\mathbb{R}^3}\varepsilon_{xc}(\rho)\rho(r)\mathrm{d}r,
\]
where $\varepsilon_{xc}(\rho)$ is the exchange-correlation functional used to
describe the 
quantum interaction between electrons \cite{yang-meza-wang07}.
Let $V_{ext}(r)$ be the external potential, the associated external
potential energy is
\[
E_{ext}=\int_{\mathbb{R}^3}V_{ext}(r)\rho(r)\mathrm{d}r.
\]
Therefore, the total energy of the system is
\begin{equation}
\begin{aligned}\label{eq:EKS}
  E_{KS}(\{\phi_i\})=
  T_{KS}+E_{H}+E_{xc}+E_{ext}+E_{II},
\end{aligned}
\end{equation}
where $E_{II} = \displaystyle\frac1{2}\sum_{i,j=1, i\neq j}^M\frac{Z_iZ_j}{|R_i-R_j|}$
is the nuclei interaction energy and $Z_i$ and $R_i$ are the atomic number and coordinate of the $i$-th nucleus, respectively. For more
details  we refer the reader to \cite{martin04}. Then finding the ground state
energy of the system is equivalent to solving the following minimization problem:
\begin{equation}\label{eq:KSinf}
  \begin{aligned}
   E_{KS}^{0}= \inf_{\phi_i\in
  H^{1}(\mathbb{R}^3) } & \quad  E_{KS}(\{\phi_i\}) \\
       \st \quad & \int_{\mathbb{R}^3}\phi_i(r)\phi_j(r)\mathrm{d}r=
  \delta_{ij},\quad 1\leqslant i,j \leqslant N.
  \end{aligned}
\end{equation}
Alternatively, one  solves the so-called KS equation
\begin{align}\label{eq:KS}
\begin{cases}
  \mathcal{H}_{KS}\phi_i\triangleq\left(-\displaystyle\frac{1}{2}
  \Delta+V_{ext}(r)+
  V_H(\rho)+V_{xc}(\rho)\right)\phi_i= \displaystyle\sum_{j=1}^{N}
  \lambda_{ij}\phi_j,\\
  \displaystyle\int_{\mathbb{R}^3}\phi_i(r)\phi_j(r)\mathrm{d}r=\delta_{ij},
\end{cases}
\end{align}
where $\displaystyle V_{xc}(\rho)=\frac{\delta E_{xc}(\rho)}{\delta\rho}$,
$i,j=1,2,\cdots,N$. The KS equation \eqref{eq:KS}
actually corresponds to the first-order necessary optimality condition of
\eqref{eq:KSinf}. It is a nonlinear eigenvalue
problem since the Hamiltonian operator $\mathcal{H}_{KS}$ is
a nonlinear operator with respect to the charge density $\rho$.

Although the KSDFT model has been very successful in many aspects, there
still exists a few difficulties. First, it can be computationally demanding
because of the determination of the $N$ wave functions whose cost is dominant.
Second, the wave functions usually oscillate rapidly near the nuclei
in realistic systems. Consequently, pseudopotential algorithm is
always introduced. Third, approximations are needed for the unknown
analytic formula of $V_{xc}(\rho)$ \cite{perdew-burke96, perdew-wang92}.


\subsection{The OFDFT model}

The  OFDFT model adopts the kinetic energy density functional (KEDF) as the kinetic energy
term instead of the one in KSDFT. Hence, $N$ independent
single-electron wave functions are no longer needed in OFDFT.
One of the most successful orbital-free models is the so-called
Thomas-Fermi-von Weizas\"{a}cker (TFW) model, whose kinetic energy
has the following representation
\begin{align}\label{eq:TFW}
T_{TFW}(\rho) = C_{TF}T_{TF}(\rho)+\mu T_{vW}(\rho),
\end{align}
where $C_{TF}=\displaystyle \frac3{10}(3\pi^2)^{\frac2{3}}$, $T_{TF}(\rho)
=\displaystyle\int_{\mathbb{R}^3}\rho^{\frac5{3}}(r)\mathrm{d}r$, $T_{vW}(\rho)=
\displaystyle\frac1{8}\int_{\mathbb{R}^3}\frac{|\nabla\rho|^2}{\rho}\mathrm{d}r$
and $\mu$ is the parameter that obtained based on physical experiments
or theoretical analysis. Since TFW model considers the gradient of the
charge density, it can deal with heterogeneous systems, especially for
diatomic systems.
A few other KEDFs were proposed in
\cite{wang-teter92, wang-govind-carter99} and references therein
in order to satisfy the linear response theory.
They share a common formula
\begin{align}\label{eq:reponseTFW}
T_{LR}(\rho)=T_{TF}(\rho) + \mu T_{vW}(\rho) + C_{TF}
\int_{\mathbb{R}^3}\int_{\mathbb{R}^3}K(r-r')\rho^{\alpha}(r)
\rho^{\beta}(r')\mathrm{d}r\mathrm{d}r',
\end{align}
where $K(r-r')$ is chosen such that $T_{LR}(\rho)$ satisfies the Lindhard
susceptibility function but using different
parameters $\alpha$ and $\beta$.
By abuse of notation, we use $T_{OF}(\rho)$ to denote either $T_{TFW}(\rho)$
or $T_{LR}(\rho)$. Subsequently, the total energy can be written as
\begin{align}\label{eq:OFeng}
E_{OF}(\rho)=T_{OF}(\rho)+E_{ext}(\rho)+E_{H}(\rho)+E_{xc}(\rho)+E_{II},
\end{align}
where $E_{ext}(\rho)$, $E_{H}(\rho)$, $E_{xc}(\rho)$ and $E_{II}$ are defined
the same as in \eqref{eq:EKS}.
Since the charge density $\rho$  should be properly
normalized to  the number of electrons $N$ in the system, finding the ground
state energy of the system can be formulated as
\begin{align}\label{eq:OFinf}
E_{OF}^0=\inf\left\{E_{OF}(\rho):\rho\in L^1(\mathbb{R}^3),\
\rho^{\frac1{2}}\in H^1(\mathbb{R}^3),\  \rho\geq 0,\
\int_{\mathbb{R}^3}\rho(r)\mathrm{d}r = N\right\}.
\end{align}
The nonnegative constraints can be eliminated by substituting $ \rho$ by
another variable $\varphi^2$. The first-order necessary optimality condition
of \eqref{eq:OFinf} with respect to KEDF \eqref{eq:TFW} is the
nonlinear eigenvalue problem
\begin{align}\label{eq:OF_euler}
\begin{cases}
\mathcal{H}_{OF}\varphi\triangleq\left(-\displaystyle\frac{\mu}{2}\Delta+
\displaystyle\frac{\delta\Big(T_{OF}(\rho)-\mu T_{vW}(\rho)\Big)}{\delta\rho}+
V_{eff}(\rho)\right)\varphi = \lambda\varphi,\\
\displaystyle\int_{\mathbb{R}^3}\varphi^2 = N,
\end{cases}
\end{align}
where $\displaystyle\frac{\delta T_{vW}(\rho)}{\delta\rho}=-\frac1{2}
\frac{\nabla^2\varphi}{\varphi}$ and
$V_{eff}(\rho)=V_{ext}(r)+V_{H}(\rho)+V_{xc}(\rho)$.

Therefore, OFDFT model expresses the system by only using the
charge density as its variable in the spirit of
Hohenberg-Kohn theorem and avoids computing $N$ eigenpairs
\cite{carling-carter03}. Numerical experiments have shown
that OFDFT is good at simulating systems with main group
elements and nearly-free-electron-like metals with comparable
accuracy to KSDFT \cite{ho-huang-carter08, wang-govind-carter99}.
However, the accuracy of OFDFT relies on the accuracy of KEDFs and the
results may not be accurate for covalently bonded and ionic systems \cite{ho-huang-carter08}.


\subsection{The discretized DFT models}\label{subsec:discretize}

In this subsection, we describe the real space discretization for KSDFT and OFDFT,
respectively. Since going into the detailed discretizing steps would take us too far
afield, we focus on a high level description without loss of generality.
The readers are referred to \cite{chen-gong-he-yang-zhou13, chen-gong-zhou10, chen-he-zhou11,
zhou07} for theoretical analysis of these discretization schemes.


Under a chosen finite difference discretization scheme, the $N$  independent
wave functions can be represented by a matrix
\[
X=[x_1,\cdots,x_N]\in\mathbb{R}^{n\times N},
\]
where $n$ is the spatial degrees of freedom of the computational domain
$\Omega$, $\{x_i\}$ are column vectors used to approximate the wave
functions $\{\phi_i\}$ and they satisfy the orthogonality constraint
\[
X^\top X=I_N,
\]
where $I_N$ is the $N\times N$ identity matrix. The discretized charge
density associated with these $N$ occupied states can be expressed as
\[
\rho(X)=\mathrm{diag}(XX^\top),
\]
where
$\mathrm{diag}(A)$ is a column vector consisting of
diagonal entries of the matrix $A$. Hence, the discretized total energy
functional  of \eqref{eq:EKS} is
\begin{equation}
  E_{KS}(X) = \frac{1}{2}\tr(X^\top LX)+\tr(X^\top EX)+
  \frac{1}{2}\rho(X)^\top L^{\dag}\rho(X)+
  \rho(X)^\top \epsilon_{xc}(\rho(X)),
\end{equation}
where the Hermitian matrix $L\in\mathbb{R}^{n\times n}$ is an approximation to
the Laplacian operator, the diagonal matrix $E$ is an approximation of $V_{ext}$, and the discretized form of the Hartree potential can be
represented by the product of Hermitian matrix $L^\dag$ with $\rho(X)$,
where $L^\dag$ is the inverse of the discretized Laplacian operator.

Therefore, the discretized minimization problem
\eqref{eq:KSinf} is
\begin{equation}\label{eq:disKSinf}
  \begin{aligned}
    \min_{X\in \mathbb{R}^{n\times N}} E_{KS}(X),\
    \quad\st\ \  X^\top X=I_N.
  \end{aligned}
\end{equation}
Using the notation $((E_{KS})_X)_{i,j}= \frac{\partial E_{KS}}{\partial X_{ij}}$ for
the partial derivative, we obtain
\begin{equation}\label{eq:partial deriv}
  (E_{KS})_X = H(X)X,
\end{equation}
where $H(X)$ is the discretized Hamiltonian matrix
\begin{equation}
  H(X)=\frac{1}{2}L+E+\mathrm{Diag}(L^\dag\rho(X))+\mathrm{Diag}(\gamma_{xc}),
\end{equation}
where $\mathrm{Diag}(v)$ denotes a diagonal matrix with $v$ on its diagonal,
and
\[
\gamma_{xc}(\rho)=\frac{\mathrm{d}[\rho\epsilon_{xc}(\rho)]}{\mathrm{d}\rho}.
\]
Consequently, the first-order necessary optimality condition of problem
\eqref{eq:disKSinf} is
\begin{equation}\label{eq:disKS}
  \begin{aligned}    H(X)X&=X\Lambda,\\  X^\top X&=I_N,\end{aligned}
\end{equation}
where $\Lambda$ is the Lagrangian multiplier which is a symmetric
matrix. Hence,  \eqref{eq:disKS} is exactly the discretization
of the continuous KS equation \eqref{eq:KS}.

We choose the finite element method to discretize the OFDFT model.
 Assume that the finite element basis sets under a shape-regular tetrahedral conforming mesh of $\Omega$ are $\{u_k\}^{n}_{k=1}$,
where $n$ is the number of the basis functions. The wave
function $\varphi$ can be approximated  by a linear combination of these
finite element basis
\begin{align}\label{eq:FEwavefunction}
\varphi(c)=\sum^n_{t=1}c_tu_t,
\end{align}
where $c=(c_1, c_2, \cdots, c_n)^\top$. The corresponding discretized
charge density can be expressed as
\[
\rho(c)=\varphi^2(c). 
\]
The discretization of the nonlinear eigenvalue problem \eqref{eq:OF_euler} is
\be \label{eq:DiscreteOF}
\int_\Omega u_k\mathcal{H}_{OF}\varphi(c)\mathrm{d}r = \lambda\int_{\Omega}
u_k\varphi(c)\mathrm{d}r,\ \ 1\leq k\leq n, 
\ee
 Substituting
\eqref{eq:FEwavefunction} into \eqref{eq:DiscreteOF} and running over
all the indices, we obtain the following generalized eigenvalue problem
\begin{align}
H(c)c = \lambda Bc,
\end{align}
where the components of $H(c)$ and $B$ are
\[
H_{kl}(c)=\int_{\Omega}u_k\mathcal{H}_{OF}u_l\mathrm{d}r,\quad
\mathrm{and}\quad B_{kl}=\int_{\Omega}u_ku_l\mathrm{d}r.
\]
The corresponding discretization  of \eqref{eq:OFinf} is
\begin{align}\label{eq:disOFinf}
\min_{c\in\mathbb{R}^{n}} E_{OF}(\rho(c)), \quad
\st\ \  c^\top Bc = 1.
\end{align}

\subsection{The self-consistent field (SCF) iteration}\label{subsec:scf}
The nonlinear eigenvalue formulations of KSDFT and OFDFT presented in subsection
\ref{subsec:discretize} give rise to the SCF iteration
naturally. Without loss of generality, we describe the SCF iteration for
the discretized KS equation \eqref{eq:disKS}. Starting from $X_0$ with $X_0^\top X_0=I$, the basic SCF iteration
computes the $(k+1)$-th iterate $X_{k+1}$ as the solution of the linear eigenvalue
problem:
\be \label{eq:SCF} \begin{aligned}  H(X_k) X_{k+1} &=  X_{k+1} \Lambda_{k+1}, \\
  X_{k+1}^\top X_{k+1} &=  I.
  \end{aligned} \ee
The convergence of the SCF iteration can often be speeded up by the so-called charge density or potential mixing techniques.
The only difference is that the coefficient matrix $H(X_k)$  in the linear
eigenvalue problem \eqref{eq:SCF} is replaced by another matrix
$H$, which is constructed from a linear combination of the previously
computed charge densities or potentials and the one obtained from certain
 schemes at current iteration. Frequently used schemes include Anderson's mixing
\cite{anderson65}, Pulay's mixing (or DIIS) \cite{pulay80} and Broyden mixing \cite{johnson88}.
We outline the major steps of the
SCF iteration in Algorithm \ref{alg:SCF}.

\begin{algorithm2e}[H]\caption{The SCF Iteration}\label{alg:SCF}
Given an initial guess on the charge density $\rho^0_{in}$.\\


Solve the linear eigenvalue problem
\[ H(\rho_{in}^k ) X_{k+1}=X_{k+1}\Lambda_{k+1}, \quad X_{k+1}^\top X_{k+1} = I_N,\] where
  $H(\rho_{in}^k )$ is the Hamiltonian at $\rho_{in}^k$, $X_{k+1}$ are
the eigenvectors corresponding to the $N$-smallest eigenvalues of $H(\rho_{in}^k
)$.\\

Compute the new charge density $\rho_{out}^{k+1}$ from $X_{k+1}$. 
Stop if a certain termination rule is met.

Compute a new
charge density $\rho_{in}^{k+1}$ using  $\rho_{out}^{k+1}$ and
   a prior chosen  density mixing scheme and go to step 2.
 \end{algorithm2e}

 Although the SCF iteration with charge density or potential mixing often works well on many
problems, it is well known that there is no theoretical guarantee on its
convergence and it can converge slowly or even fail on certain problems \cite{yang-meza-wang07}
. On the other hand, the main
computational bottleneck of the SCF iteration is solving a sequence of linear eigenvalue
problems.
 Since the degrees of freedom $n$ is usually very large, in particular
for large-scale systems, the computational cost of eigenpairs is the
dominant factor of the SCF iteration.

\section{Gradient type algorithms for total energy minimization}\label{sec:alg}

 The minimization models \eqref{eq:KSinf} and \eqref{eq:OFinf} provide many alternative
methodologies other than solving linear eigenvalue problems repeatedly.
Since $B$ in \eqref{eq:disOFinf} is positive
definite,  \eqref{eq:disOFinf}  can be transformed to a minimization
problem under  a unit spherical constraint
as
\be \label{eq:OFDFT-y} \min_{y\in\mathbb{R}^{n}} E_{OF}(\rho(B^{-\frac{1}{2}}y)),\ \
\st\ \  \|y\|_2 = 1.
\ee
Therefore, both the KSDFT model \eqref{eq:disKSinf} and the OFDFT model \eqref{eq:OFDFT-y}  are unified under optimization with
orthogonality constraints as
\begin{equation}\label{opt}
\min_{X\in\mathbb{R}^{n\times p}} E(X),\ \
\st\ \  X^\top X = I_p,
\end{equation}
where $E(X)$ is the corresponding total energy functional  of either KSDFT or
OFDFT models and $p$ is the number
of eigenpairs needed to be calculated. Inspired
by the preliminary but promising performance of the feasible method
for \eqref{opt}  in \cite{zai-yin},
we further investigate  gradient type methods for solving OFDFT and KSDFT,
including one approach using QR factorization \cite{opt-manifold-book}. Each trial point generated by these methods is guaranteed
to be on the unit sphere or satisfy the orthogonality constraints.
These schemes are numerically efficient and let us apply state-of-the-art
acceleration techniques such as Barzilai-Borwein steps and non-monotone
line search with global convergence guarantees.

 Let $\Mnp=\{X\in\mathbb{R}^
{n\times p}: X^\top X = I\}$  be the feasible set, which is often referred to the
Stiefel manifold. When $p=1$, it reduces to the unit-sphere manifold $S^{n-1}=\{x\in\mathbb{R}^n:
\| x \|_2= 1\}$. The tangent space of $\Mnp$ at $X$ is \[ \T_X \Mnp=\{ Z \in \Rnp : X^\top Z
+ Z^\top X = 0 \}.\]
The gradient $\nabla E(X)$   at $X$ is defined as the element of  $ \T_X \Mnp$
that satisfies
\be \label{eq:def-grad}  \tr(E^\top_X \Delta )=\iprod{ \nabla E}{\Delta }_X
 := \tr\left(\nabla E^\top (I- \frac{1}{2}  X X^\top) \Delta \right), \quad
 \mbox{ for all } \Delta \in \T_X \Mnp,  \ee
 where $E_X$ is the partial derivative of $E(X)$, i.e.,
 $(E_X)_{ij} = \Big(\mathcal{D}E(X)\Big)_{ij} = \displaystyle\frac{\partial E_{KS}}{\partial
 X_{ij}}$. For the KSDFT and OFDFT models, it holds
 \be \label{eq:deriv-E-X} E_X = H(X) X,\ee
 where $H(X)$ is the corresponding Hamiltonian matrix.
 An elementary verification shows that
\be\label{eq:def-nablaE}
\nabla E(X) = E_X - X E_X^\top X.
\ee
The first-order optimality conditions of \eqref{opt} are
\[ \nabla E(X) = 0, \mbox{ and } X^\top X = I. \]

Our proposed algorithms are adapted from the classical steepest descent
method. The orthogonality constraints are preserved at a reasonable computational cost.
 Since $\nabla E$ is the gradient on the manifold, a natural idea is to compute
the next iterates  as
\be\label{eq:def-Y}
Y = X - \tau \nabla E(X),
\ee
where $\tau$ is a step size. The obstacle is that the new point $Y$ may not satisfy
$Y \in \Mnp$.
 We next describe two approaches for overcoming this difficulty.

Our first strategy uses the constraint-preserving scheme proposed in
 \cite{zai-yin} by slightly modifying the term $\nabla E(X)$.
Suppose that $X^\top X =I$ and define $W$ as a skew-symmetric matrix
\[
W:= E_X X^\top - X E_X^\top,
\]
which yields $\nabla E(X) =  W X$. Then the new trial point is generated as
  \be \label{eq:WY} X_\WY(\tau) = X -\tau W
  \left(\frac{X+X_{\mathrm{WY} }(\tau)}{2}\right). \ee
It can be shown that  $X_\WY(\tau)$ is a orthogonal matrix
and it  defines a projected gradient-like curve on the Stiefel manifold.
 \begin{lemma}[Lemmas 3 and 4 in \cite{zai-yin}] \label{lemma:skew-symm}
1) The matrix $X_\WY(\tau)$ defined by  \eqref{eq:WY}  can be expressed as  \be
\label{eq:WY1} X_\WY(\tau) = \left(I +\frac{\tau}{2} W \right)^{-1}\left(I - \frac{\tau}{2} W \right)X , \ee
which satisfies $X_\WY(\tau)^\top X_\WY (\tau) = X^\top X$
and $X_\WY'(0) = - \nabla E(X)$.

2) Rewrite $W= U V^\top$ for $U = [E_X, \; X]$ and $ V = [X, \; -E_X]$. If $ I +
\frac{\tau}{2} V^\top  U$ is invertible,  then \eqref{eq:WY1} is
equivalent to
\be \label{eq:WY3}
X_\WY(\tau) =X - \tau U \left( I + \frac{\tau}{2} V^\top  U \right) ^{-1} V^\top X.
\ee
3) Suppose $p=1$ and $ W=a x^\top - x a^\top$, where $a = E_x$.
Then \eqref{eq:WY1} is given explicitly by
\be\label{eq:WY-1d}
x_\WY(\tau)= x - \beta_1(\tau) a - \beta_2(\tau)x,
\ee
where $\beta_1(\tau) = \tau \frac{  x^\top x   } { 1- \left(
\frac{\tau}{2}   \right)^2 (a^\top x)^2 + \left( \frac{\tau}{2}   \right)^2
\|a\|^2_2 \|x\|^2_2 }$ and $ \beta_2(\tau) =   -\tau \frac{  x^\top a +
\frac{\tau}{2} \left( (a^\top x)^2 - (a^\top a) (x^\top x) \right)  }
{ 1- \left( \frac{\tau}{2}   \right)^2 (a^\top x)^2 +
\left( \frac{\tau}{2} \right)^2 \|a\|^2_2 \|x\|^2_2 }$.
\end{lemma}

Using the convention that an $m\times p$ matrix times a $p\times n$
matrix costs $2mnp$ flops, we calculate the computational complexity
of the scheme \eqref{eq:WY} as below. The cost of the scheme \eqref{eq:WY}
reduces to two inner products in case of $p=1$.
When $1<p \ll n$, the formula \eqref{eq:WY3} should be used to
compute $X_\WY(\tau)$ whose cost is $7np^2+\frac{40}{3}p^3+O(np)$.
This number shows that the complexity depends on
both the spatial degrees of freedom $n$
 and the number of columns $p$. Although orthogonality is preserved
 theoretically, we observe in our
 numerical experiments that it may lose due to numerical errors and a
 reorthogonalization step is needed. Further analysis on controlling the errors can
 be found in \cite{jiang-dai}.

 Our second strategy is to orthogonalize
$Y$ explicitly by using the QR factorization
\be\label{eq:QR}
X_\QR(\tau)=\mathrm{qr}(Y),
\ee
where $\mathrm{qr}(Y)$ is the column-orthogonal matrix $Q$ corresponding to the
QR factorization of $Y=QR$. Therefore, the update scheme \eqref{eq:QR}
can be viewed as a kind of projected gradient method on the Stiefel manifold.
 The next proposition shows that the matrix $Y$  in \eqref{eq:def-Y}  is always full rank  and the condition number of
matrix $Y^\top Y$ is bounded if the step size $\tau$ and the norm of  $\|H(X)\|_2$ are bounded.

\begin{Proposition}\label{prop:Y-fullrank}
  Suppose that $X\in\mathbb{R}^{n\times p}$ satisfies $X^\top X=I$. Then the
  matrix $Y$
   computed  by \eqref{eq:def-Y} is full rank for any $\tau \in \mathbb{R}$ and
   the eigenvalue of $Y^\top Y$ is bounded as
  \begin{equation}\label{eq:YYeigenRange}
    1 \leqslant \lambda(Y^\top Y) \leqslant \tau^2\|H(X)\|_2^2+1.
  \end{equation}
\end{Proposition}
\begin{proof}
Since $Y=X-\tau \nabla E(X)$, we have
  \begin{align*}
    Y^\top Y 
             & = X^\top X - \tau(\nabla E(X)^\top X + X^\top\nabla E(X)) + \tau^2\nabla E(X)^\top\nabla E(X).
  \end{align*}
  The definition of $\nabla E(X)$ in \eqref{eq:def-nablaE} and $X^\top X=I$ yield
  \begin{align*}
    \nabla E(X)^\top X + X^\top\nabla E(X) &= (E_X^\top - X^\top E_XX^\top)X + X^\top(E_X - XE_X^\top X),\\
                                           &= E_X^\top X -X^\top E_X+X^\top E_X - E_X^\top X=0,
  \end{align*}
 which further gives
  \[
  Y^\top Y = I + \tau^2 \nabla E(X)^\top \nabla E(X).
  \]
Since $\nabla E(X)^\top \nabla E(X)$ is positive semidefinite and its largest
  eigenvalue is $\|\nabla E(X)\|^2_2$, we obtain
 \be\label{Y-fullrank-1}
  1 \leqslant \lambda(Y^\top Y) \leqslant \tau^2\|\nabla E(X)\|_2^2+1.\ee
Substituting the expression of $E_X$ in \eqref{eq:deriv-E-X} into $\nabla E(X)$
and using the fact $I-XX^\top$ is a projection, we have
\be \label{Y-fullrank-2}
\|\nabla E(X)\|_2 \leqslant \|(I-XX^{\top})H(X)X)\|_2 \le \|H(X)\|_2.
\ee
Combining \eqref{Y-fullrank-1} and \eqref{Y-fullrank-2} together proves \eqref{eq:YYeigenRange}.
\end{proof}

There are many approaches for computing the QR factorization.
 Propositon \ref{prop:Y-fullrank} implies that the matrix $Y^\top Y$ is
 well-conditioned under a suitable chosen step size
$\tau$. Hence, the QR factorization based on the Cholesky
 factorization can be computed stably and accurately, for example, using the
 efficient implementation in LAPACK. Specifically,  the matrix $Q$ can be
 assembled as in Algorithm  \ref{alg:ChQR}.


\begin{algorithm2e}[H]\caption{Cholesky QR factorization}
\label{alg:ChQR}
Input the matrix $Y$.\\
Compute the Cholesky factorization $L L^\top$ of $Y^\top Y$. \\
Output $Q=Y L^{-1}$.
\end{algorithm2e}


Using the same notation as our analysis for the complexity of scheme
\eqref{eq:WY}, the cost of scheme \eqref{eq:QR} is $6np^2+\frac1{3}p^3+O(np)$
since computing the gradient $\nabla E(X)$ and the
QR-factorization need $4np^2$ and $2np^2+\frac{1}{3}p^3+O(np)$,
respectively. We can see that the QR-based method is slightly cheaper than the first strategy.

Another critial algorithmic issue is the determination of a suitable step size
$\tau$. Instead of using the classical Armijo-Wolfe based monotone line search, we
apply  the nonmonotone  curvilinear (as our search path is on the manifold rather
than a straight line) search with an initial step size determined by the
Barzilai-Borwein (BB) formula, which we have found more efficient for our problem.
They were developed originally for the vector case in \cite{BarzilaiBorwein1988}.
At iteration $k$, the step size is computed as
 \be  \label{eq:bb-1} \tau_{k,1} = \frac{\tr\left((S_{k-1})^{\top}S_{k-1}\right)}
    {|\tr\big((S_{k-1})^{\top}  Y_{k-1} \big)|} \quad \mbox{ or  } \quad
    \tau_{k,2} = \frac{|\tr\left((S_{k-1})^{\top} Y_{k-1}
    \right)|}{ \tr \big((Y_{k-1})^{\top} Y_{k-1} \big)},\ee
where
$ S_{k-1} = X_k - X_{k-1}$ and $Y_{k-1} = \nabla E(X_k) -\nabla E(X_{k-1})$.
In order to guarantee convergence, the final value for $\tau_k$ is a fraction
(up to 1, inclusive) of $\tau_{k,1}$ or $\tau_{k,2}$ determined by a
nonmonotone search condition. Let $X(\tau)$ be either of \eqref{eq:WY} or
\eqref{eq:QR}, $C_0=E(X_0)$, $ Q_{k+1} = \eta Q_k +1$ and  $Q_0=1$. The new points
 are generated iteratively in the form  $X_{k+1}:=X_k(\tau_k)$, where
  $\tau_k = \tau_{k,1} \delta^h$ or $\tau_k = \tau_{k,2} \delta^h$   and
  $h$ is the smallest nonnegative integer satisfying
 \be \label{eq:NMLS-Armijo}
E(X_k(\tau_k)) \le C_k - \rho_1  \tau_k \|\nabla E(X_k)\|_F^2, \ee
 where each reference
 value $C_{k+1}$   is taken to be the convex combination of  $C_k$ and
 $E(X_{k+1})$ as $C_{k+1} = (\eta Q_k C_k +
            E(X_{k+1}))/Q_{k+1}$. 
In Algorithm \ref{alg:ConOptM} below, we specify our method for solving the DFT
models. Although several backtracking steps may be needed to update the
$X_{k+1}$, we observe that the BB step size  $\tau_{k,1}$ or $\tau_{k,2}$ is
often sufficient for \eqref{eq:NMLS-Armijo} to hold in most of our numerical
experiments. In the case that  $\tau_{k,1}$ or $\tau_{k,2}$ is not bounded, they
are reset to a finite number and convergence of our algorithm still holds.

\begin{algorithm2e}[H]\caption{Constraint Optimization on Stiefel Manifold}
\label{alg:ConOptM}
Given $X_0$, set $\tau>0$,  $ \rho_1, \delta, \eta, \epsilon \in (0,1)$, $k=0$.\\
\While{$\|\nabla E(X_k)\|>\epsilon$ }{
Compute $\tau_k \gets \tau_{k,1} \delta^h$ or $\tau_k \gets \tau_{k,2} \delta^h$, where
$h$ is the smallest nonnegative integer satisfying the condition \eqref{eq:NMLS-Armijo}.\\
Set $X_{k+1}\gets X_\QR(\tau)$ \mbox{ or } $X_{k+1}\gets X_\WY(\tau)$. \\
{$Q_{k+1}\gets\eta Q_k +1$ and $C_{k+1}\gets(\eta Q_k C_k
+ E(X_{k+1}))/Q_{k+1}$}.\\
  $k\gets k+1$.
}
Calculate the ground state energy and other physical quantities.
\end{algorithm2e}

We next summarize the computational
complexity of Algorithm \ref{alg:ConOptM} with respect
to schemes \eqref{eq:WY} and \eqref{eq:QR}, respectively.
Each iteration of \eqref{eq:WY} has a minimal complexity of $9np^2+\frac{40}{3}p^3+O(np)$
since the cost of computing
$X_{WY}(\tau)$ is $7np^2+\frac{40}{3}p^3+O(np)$ and the
assembling of the gradient $\nabla E(X)$ for the BB step size needs another $2np^2$.
The work for a different $\tau$ is $4np^2+O(p^3)$ because of
the saving of some intermediate variables.
On the other hand, the minimal cost of each iteration of the QR-based method
\eqref{eq:QR} is $6np^2+\frac1{3}p^3+O(np)$. The cost for a
different $\tau$ during backtracking line search is that of a new QR-factorization because $\nabla E(X)$ is available.
 When $p$ is larger than a
few hundreds, the inversion of $I + \frac{\tau}{2} V^\top  U$ in \eqref{eq:WY3}
is not negligible and the LU decomposition is usually more expensive than
the Cholesky factorization.


We make the following assumption for the convergence of our gradient type methods.
\begin{Assumption}\label{assump:lipschitz}
The total energy function $E(X)$ is differentiable and its derivative $ E_X(X)$ is Lipschitz continuous
with Lipschitz constant $L_0$, i.e.,
\[
\| E_X(X) - E_X(Y)\|_F \leq L_0\|X-Y\|_F, \  \text{for all}
\  X,\ Y \in \mathcal{M}^p_n.
\]
\end{Assumption}

Although Assumption
\ref{assump:lipschitz} may not be satisfied in many cases due to the
exchange-correlation term, it holds in cases such as the
Gross-Pitaevskii equation \cite{zhou04}.
 Using the proofs of \cite{jiang-dai} in a similar fashion, we can establish the
  convergence of Algorithm \ref{alg:ConOptM} as follows.
\begin{Theorem}\label{theo:convergence}
Suppose that Assumption
\ref{assump:lipschitz} holds.  Let $\{X_k \mid k\geq 0\}$ be a sequence generated by Algorithm \ref{alg:ConOptM}
 using $\epsilon = 0$, $\tau_k = \tau_{k,1} \delta^h$ and $\rho_1 < \frac1{2}$. Then the step size satisfies
\[
\tau_k\geq \min\{c,\tau_{k,1}\},
\]
where $c$ is some constant.  Furthermore, either  $\|\nabla E(X_k)\| = 0$ for some finite
$k$ or
\[
\liminf_{k\rightarrow\infty}\|\nabla E(X_k)\|_F = 0.
\]
\end{Theorem}

We should point out that both schemes
\eqref{eq:WY} and \eqref{eq:QR} are special cases of
  optimization on manifold in \cite{opt-manifold-book}.
 A map $\cR: \T_X\mathcal{M}\to \mathcal{M}$ is called a retraction if
     \begin{enumerate}
      \item $\cR(0_X) = X$, where $0_X$ is the origin of $\T_X$.
      \item $\ddt \cR(tU)|_{t=0} = U $, for all $U \in \T_X$.
    \end{enumerate}
    It can be verified that both schemes \eqref{eq:WY} and \eqref{eq:QR} are
    retractions. They map a tangent vector of $\Mnp$ at $X$ to a member on $\Mnp$.
 Global convergence of the algorithms using monotone line search schemes can be obtained under some mild conditions
 \cite{opt-manifold-book}.   There are many other types of retractions. They can
 be applied to solve OFDFT and KSDFT as long as their computational cost is not
 expensive.

The discretization  on a  fine
mesh for large-scale systems usually leads to a  problem of huge size whose
computational cost is expensive. A useful technique is adaptive mesh refinement, where the discretized problems
are solved in turn from the coarsest mesh to the finest mesh and the starting point at each level other than
the coarsest is obtained by projecting the solution obtained on the previous
(i.e., next coarser) mesh.  We present our adaptive mesh refinement method in
Algorithm \ref{alg:AdpConOptM}.

\begin{algorithm2e}[H]\caption{Adaptive Mesh Refinement Method on
Stiefel Manifold}\label{alg:AdpConOptM}

Given an initial mesh $\mathcal{T}_0$ and initial wavefunctions $\hat{X}^0$. Set $i=1$.

Use $\hat{X}^{i-1}$ as an initial guess
on the $i$-th  mesh $\mathcal{T}_i$ to calculate the ground state
wavefunctions $\hat{X}^i$ using Algorithm \ref{alg:ConOptM}.  \label{step:AdpConOptM2}

Calculate a posteriori error estimator $\eta_i$ on mesh $\mathcal{T}_i$.
If $\eta_i<\epsilon$, evaluate physical quantities and stop. Otherwise,
mark and refine the mesh to obtain $\mathcal{T}_{i+1}$, and go to step
\ref{step:AdpConOptM2}.

\end{algorithm2e}


\section{Numerical experiments} \label{sec:num}
We now demonstrate the  efficiency and robustness of our
gradient type methods for solving both KSDFT and OFDFT models. All experiments
are performed on a PC cluster LSSC-III in the State Key Laboratory
of Scientific and Engineering Computing, Chinese Academy of Sciences.
Each node of LSSC-III  contains two Intel X5550 GPUs and 24GB memory.
Our implementation is parallelized by using MPI. Throughout our numerical
experiments, we use Troullier-Martins norm conserving pseudopotentials
\cite{troullier-martins1991} and choose
local density approximation (LDA) to approximate $V_{xc}(\rho)$ \cite{kohn-sham65}.

\subsection{Numerical results for the KSDFT model}\label{subsec: KSDFT}
In this subsection, we compare the gradient type method using \eqref{eq:WY}
(denoted by ``OptM-WY'') and the one using
\eqref{eq:QR} (denoted by ``OptM-QR'')  with the SCF iteration on KSDFT.
The source code of SCF is taken from the software Octopus (version 4.0.1) \cite{octopus},
an open source ab initio real-space
computing platform using finite difference discretization.
Both gradient type methods are implemented based on Octopus
and they use the same computational subroutines  as
SCF wherever it is possible.  The reported total energy
functional is computed according to \eqref{eq:EKS} rather than the original one
in Octopus based on precalculated eigenvalues.
However, these two formulas are equivalent mathematically. The initial guess is generated by linear combination of atomic orbits
(LCAO) method.  All three methods are terminated if
residuals of the gradient on manifold is smaller than some prescribed tolerance
$\varepsilon_g$, that is,
\[ \|E_X-XE_X^TX\|_{F} \le \varepsilon_g. \]
For SCF, the Broyden
method is used as the charge density mixing strategy. The linear eigenvalue
problem is solved by a preconditioned
conjugate gradient method (PCG) and it is terminated if the residual of the
eigenpairs is smaller than $0.1\varepsilon_g$ or the number of iterations
reaches 25.
 In fact, we have tested most eigensolvers available in Octopus, including
PCG, a new CG method developed in \cite{Jiang03}, a preconditioned Lanczos Scheme
\cite{Saad96}, and LOBPCG.   The reason of choosing PCG is  that it is one of
the best methods in our tests. Since the gradient type methods
may stagnate when the iterates are close to the solution, especially for large
scale systems, 
 we also terminate 
 if the relative change of
 the total energy functional is small, i.e.,
$ \frac{df_{k}+df_{k-1}+df_{k-2}}{3} < 10^{-13}$,
where $df_k=\displaystyle\frac{|E(X_{k-1})-E(X_k)|}{|E(X_{k-1})|+1}$
for the $k$-th iteration step.
An orthogonalization step is  executed in OptM-WY if $\|X^\top X-I\|_F > 10^{-12}$
to enforce orthogonality.

We choose eight typical molecular systems,
including benzene ($C_6H_6$), valine ($C_{5}H_{9}O_2N$), aspirin ($C_9H_8O_4$),
fullerene ($C_{60}$), alanine chain ($C_{33}H_{11}O_{11}N_{11}$),
carbon nano-tube ($C_{120}$), biological ligase 2JMO
($C_{178}H_{283}O_{50}N_{57}S$) \cite{beasley-hristova-shaw07}
and protein fasciculin2 ($C_{276}H_{442}O_{90}N_{88}S_{10}$)
\cite{bourne-taylor-marchot95}, without considering the spin
degrees of freedom.
In particular, the size of the matrix $X$ is  $1226485\times 793$ and
$1903841\times 1293$ in 2JMO and  fasciculin2, respectively. Our first experiment is performed
using the tolerance $\varepsilon_g = 10^{-6}$.  A summary of numerical results
is presented in Table \ref{table:LCAOh3}, where ``$E_0$'', ``$\Delta E_0$'' and
``$resi$''  denote the ground state energy, the relative total energy reduction
$E_0 - E_{min}$ where $E_{min}$ is a reliable minimum of the total energy and
the residual $\|E_X-XE_X^TX\|_{F}$ at the computed solution, respectively,
``$Iter$'' denotes the total number of iterations of each run,  ``cpu''
denotes the CPU time measured in seconds, and ``cores'' denotes the number of
CPU cores used in that computation. Both $E_0$ and $\Delta E_0$ are measured in
atomic unit (a.u.). We should point out that the number of
cores used in Table
\ref{table:LCAOh3} and Table \ref{table:epsilong} are chosen as $2^s$, where
$s$ is the largest integer so that  the number of elements on each processor
will not be smaller
than the value recommended  by Octopus. 
\begin{table}[h]
  \begin{center}
    \tabcolsep=8pt
    \begin{tabular}{|c|ccccc|}
      \hline
       {solver}&{ $E_0$(a.u.)}&{$\Delta E_0$(a.u.)}&{Iter}&{resi}&{cpu(s)}\\
      \hline
      \multicolumn{6}{|c|}{benzene$\quad p=15\quad n=64789$\quad{\bf $cores=8$}}\\
      \hline
      SCF     & -3.78474441e+01 & 1.02e-12 & 12  & 8.30e-07 & 7 \\
      OptM-WY & -3.78494441e+01 & 3.98e-13 & 120 & 5.83e-07 & 4 \\
      OptM-QR & -3.78494441e+01 & 5.11e-13 & 97  & 9.34e-07 & 4 \\
      \hline
      \multicolumn{6}{|c|}{valine$\quad p=23\quad n=109845$\quad{\bf $cores=8$}}\\
      \hline
      SCF     & -7.57851557e+01 & 1.32e-12 &  16 & 7.29e-07 & 17\\
      OptM-WY & -7.57851557e+01 & 1.28e-11 & 163 & 8.49e-07 & 10 \\
      OptM-QR & -7.57851557e+01 & 1.29e-12 & 211 & 1.85e-07 & 13 \\
      \hline
      \multicolumn{6}{|c|}{aspirin\quad $p$=34\quad $n$=133445\quad{\bf $cores=16$}}\\
      \hline
      SCF     & -1.20229138e+02 & 3.31e-12 & 17  & 8.61e-07 & 22 \\
      OptM-WY & -1.20229138e+02 & 4.83e-12 & 141 & 4.73e-07 & 11 \\
      OptM-QR & -1.20229138e+02 & 1.56e-12 & 152 & 7.18e-07 & 12 \\
      \hline
      \multicolumn{6}{|c|}{$C_{60}\quad p=120\quad n=191805$\quad{\bf $cores=16$}}\\
      \hline
      SCF     &	-3.42875137e+02	& 4.04e-12 & 23  &	6.51e-07 & 226 \\
      OptM-WY &	-3.42875137e+02	& 6.54e-12 & 239 &	5.68e-07 & 101 \\
      OptM-QR & -3.42875137e+02	& 3.25e-11 & 242 &	9.53e-07 &  96 \\
      \hline
      \multicolumn{6}{|c|}{alanine chain$\quad p=132\quad n=293725$\quad
      {\bf $cores=32$}}\\
      \hline
      SCF     & -4.78063923e+02 & 4.98e-01 & 200  & 4.18e-01 &  2769\\
      OptM-WY & -4.78562217e+02 & 5.82e-10 & 2082 & 9.64e-07 &  1102 \\
      OptM-QR & -4.78562217e+02 & 1.73e-10 & 1413 & 5.57e-07 &   712 \\
      \hline
      \multicolumn{6}{|c|}{$C_{120}\quad p=240\quad n=354093$\quad{\bf $cores=32$}}\\
      \hline
      SCF     & -6.84246913e+02 & 2.20e-01 &  200 & 2.89e-01 & 8159 \\
      OptM-WY & -6.84467036e+02 & 1.88e-09 & 1964 & 9.73e-07 & 2339 \\
      OptM-QR & -6.84467036e+02 & 2.06e-09 & 2062 & 9.95e-07 & 2213 \\
      \hline
      \multicolumn{6}{|c|}{2JMO$\quad p=793\quad n=1226485$\quad{\bf
      $cores=128$}}\\
      \hline
      SCF     &  7.42565784e+04 & 7.68e+04 &  200 & 3.18e+02 & 68988\\
      OptM-WY & -2.56413550e+03 & 9.98e-05 & 1521 & 4.36e-05 & 15757 \\
      OptM-QR & -2.56413550e+03 & 9.96e-05 & 1878 & 3.94e-05 & 15727 \\
      \hline
      \multicolumn{6}{|c|}{fasciculin2 $\quad p=1293\quad n=1903841$\quad{\bf
      $cores=256$}}\\
      \hline
      SCF     &  1.63686511e+05 & 1.68e+05 &  200 & 5.39e+02 & 148710\\
      OptM-WY & -4.26018878e+03 & 3.50e-05 & 2337 & 5.21e-05 &  49532 \\
      OptM-QR & -4.26018877e+03 & 4.44e-05 & 2414 & 5.93e-05 &  39102\\
      \hline
    \end{tabular}
  \end{center}
  \caption{A comparison of numerical results among different solvers on
  achieving
  $\varepsilon_g\,=\,10^{-6}$.}\label{table:LCAOh3}
\end{table}

We can observe from Table \ref{table:LCAOh3} that OptM-WY and OptM-QR are faster
than SCF on all instances, and all three methods are able to
compute solutions with residuals smaller than the given tolerance on benzene, valine, aspirin and
$C_{60}$. The SCF method fails to converge on alanine chain, $C_{120}$, 2JMO and fasciculin2  in terms of both the energy reduction and gradient
residuals. Both OptM-WY and OptM-QR are able to converge on alanine chain and $C_{120}$,
but achieve a residual in the order of $10^{-5}$ on 2JMO and fasciculin2. We
further illustrate the residuals $\|\nabla(E(X_k))\|_F$ and the energy reduction $E(X_k)-E_{min}$
of $C_{120}$ in Figure \ref{figure:gdiffC120}. Although these values oscillate
sharply without a descending trend in SCF,  they are reduced steadily in
OptM-WY and OptM-QR. 

\if false
\begin{figure}[h]
\begin{center}
\subfigure[]{\label{fig:SCFgdiff120}
\includegraphics[width=6cm]{fig/C120SCFgdiffG6.pdf}}
\subfigure[]{\label{fig:SCFediff120}
\includegraphics[width=6cm]{fig/C120SCFediffG6.pdf}}\\
\subfigure[]{\label{fig:OPTgdiff120}
\includegraphics[width=6cm]{fig/C120OPTgdiffG6.pdf}}
\subfigure[]{\label{fig:OPTediff120}
\includegraphics[width=6cm]{fig/C120OPTediffG6.pdf}}\\
\subfigure[]{\label{fig:QRgiff120}
\includegraphics[width=6cm]{fig/C120QRgdiffG6.pdf}}
\subfigure[]{\label{fig:QRediff120}
\includegraphics[width=6cm]{fig/C120QRediffG6.pdf}}
\end{center}
\caption{Residuals $\|\nabla E(X_k)\|_F$ and
the energy reduction $E(X_k)-E_{min}$ of SCF, OptM-WY and OptM-QR on the example $C_{120}$.}\label{figure:gdiffC120}
\end{figure}
\fi

\begin{figure}[h]
\begin{center}
\includegraphics[width=0.8\textwidth,height=1\textwidth]{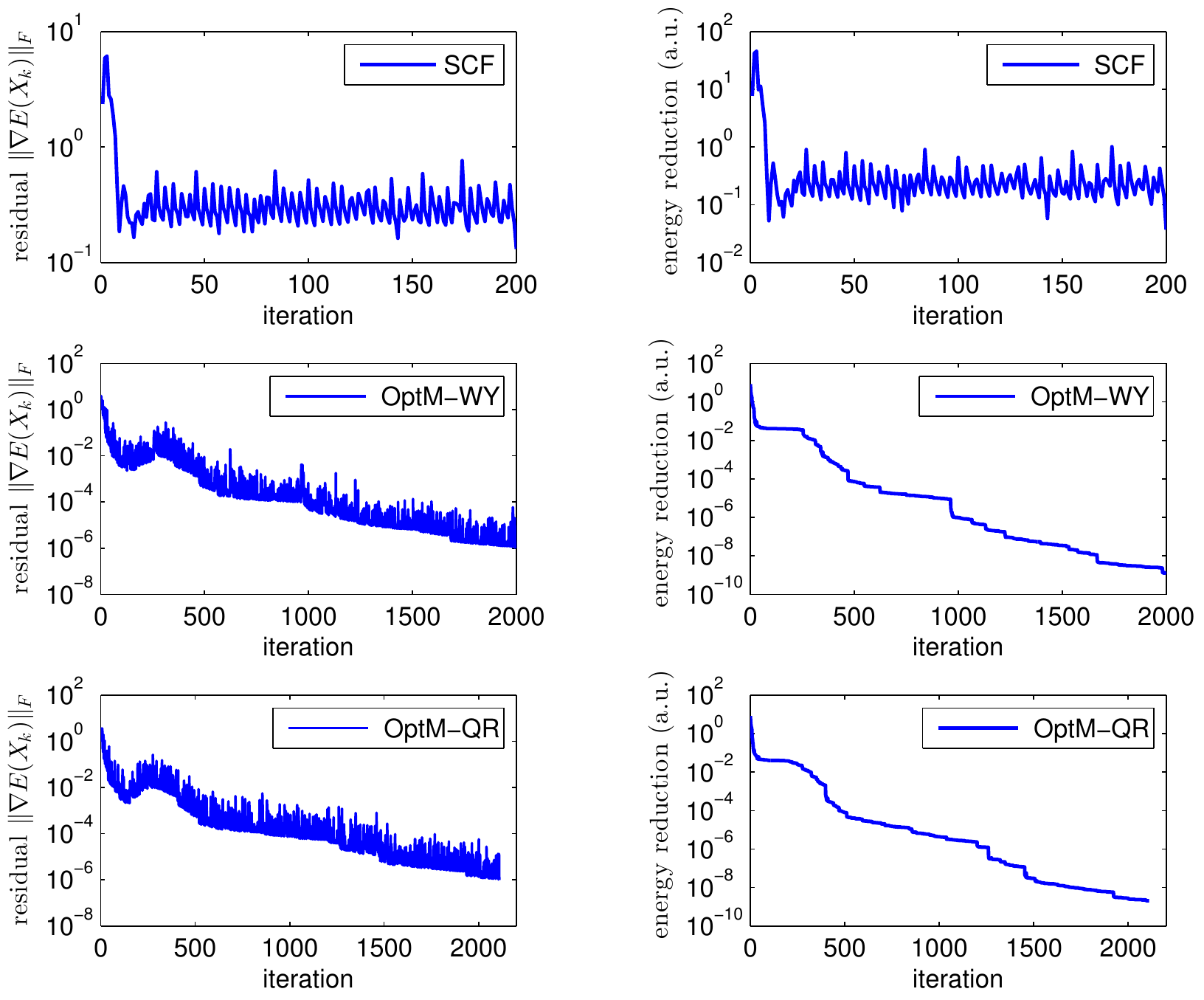}
\end{center}
\caption{Residuals $\|\nabla E(X_k)\|_F$ and
the energy reduction $E(X_k)-E_{min}$ of SCF, OptM-WY and OptM-QR on  $C_{120}$.}\label{figure:gdiffC120}
\end{figure}

Table \ref{table:LCAOh3} also shows that OptM-QR
is faster than OptM-WY on most test problems. We next illustrate their
convergence behavior using three different tolerances $\varepsilon_g=10^{-5}, 10^{-6}$ and $10^{-7}$ on alanine chain and
$C_{120}$ molecules. The results are reported in Table \ref{table:epsilong}.
 It follows from Tables \ref{table:LCAOh3} and \ref{table:epsilong} that the gradient type
methods can often attain a highly accurate solution and OptM-QR
behaves slightly better than OptM-WY, especially on large systems.
\begin{table}[ht]
\centering
    \begin{tabular}{|c|ccccc|}
      \hline
       {solver}&{\bf $\varepsilon_g$}&{\bf $\Delta E_0$}(a.u.)&{Iter}&{resi}&{cpu(s)}\\
      \hline
      \multicolumn{6}{|c|}{alanine chain$\quad p=132\quad n=293725$\quad
      {\bf $cores=32$}}\\
      \hline
      OptM-WY & 1e-05 & 1.44e-08 & 1652 & 7.80e-06 &  881 \\
      OptM-QR & 1e-05 & 3.50e-08 & 1199 & 9.30e-06 &  606 \\
      OptM-WY & 1e-06 & 5.82e-10 & 2082 & 9.64e-07 & 1102 \\
      OptM-QR & 1e-06 & 1.73e-10 & 1413 & 5.57e-07 &  712 \\
      OptM-WY & 1e-07 & 1.66e-10 & 2814 & 9.72e-08 & 1921 \\
      OptM-QR & 1e-07 & 1.65e-10 & 2016 & 9.17e-08 & 1474 \\
      \hline
      \multicolumn{6}{|c|}{$C_{120}\quad p=240\quad n=354093$\quad{\bf $cores=32$}}\\
      \hline
      OptM-WY & 1e-05 & 3.18e-08 & 1537 & 9.34e-06 & 1861 \\
      OpyM-QR & 1e-05 & 6.91e-08 & 1383 & 9.42e-06 & 1506 \\
      OptM-WY & 1e-06 & 1.88e-09 & 1964 & 9.73e-07 & 2339 \\
      OptM-QR & 1e-06 & 2.06e-09 & 2062 & 9.95e-07 & 2213 \\
      OptM-WY & 1e-07 & 1.23e-09 & 2776 & 9.90e-08 & 4433 \\
      OptM-QR & 1e-07 & 1.23e-09 & 3020 & 9.97e-08 & 4261 \\\hline
    \end{tabular}
    \caption{Numerical results with respect to different
    $\varepsilon_g$.}\label{table:epsilong}
\end{table}

We next examine parallel scalability of all three methods. For brevity, we only
show results for the systems: $C_{60}$, alanine chain, 2JMO and fasciculin2. Let
$k_0$ be the smallest number of cores so that the required memory for
the given problem can fit in these cores.
  The speedup factor for running a code on $k$ cores is defined as
  \be \label{eq:speedup-factor}
\mbox{speedup-factor}(k_0, k) = \frac{\mbox{wall clock time for a $k_0$-core run}}
{\mbox{wall clock time for a $k$-core run}}.
\ee
When the wall clock time is measured, we only run 10 iterations for SCF and 100 iterations for OptM-WY and OptM-QR
since the parallel speedup factor should not change if more iterations are preformed.
The wall clock time for each algorithm is split into two parts. The part, denoted by T0, involves
 the calculation of the total energy, gradients and Hamiltonian, which are
 determined by the specific implementation of Octopus. For example, the calculation of gradients uses the
subroutine ``Hamiltonian\_apply'' and the update of Hamiltonian uses the subroutines
``density\_calc'' and ``v\_ks\_calc''. The calculation of the total energy uses a
revision of the origin subroutine ``total\_energy'' by setting the parameter ``full''
 to be ``.ture.'' and then summing up all energy terms. All other wall
 clock time is counted as T1, which reflects the algorithmic difference among
 different algorithms.  

\begin{figure}[htb]
\begin{center}
\includegraphics[width=0.9\textwidth,height=0.7\textwidth ]{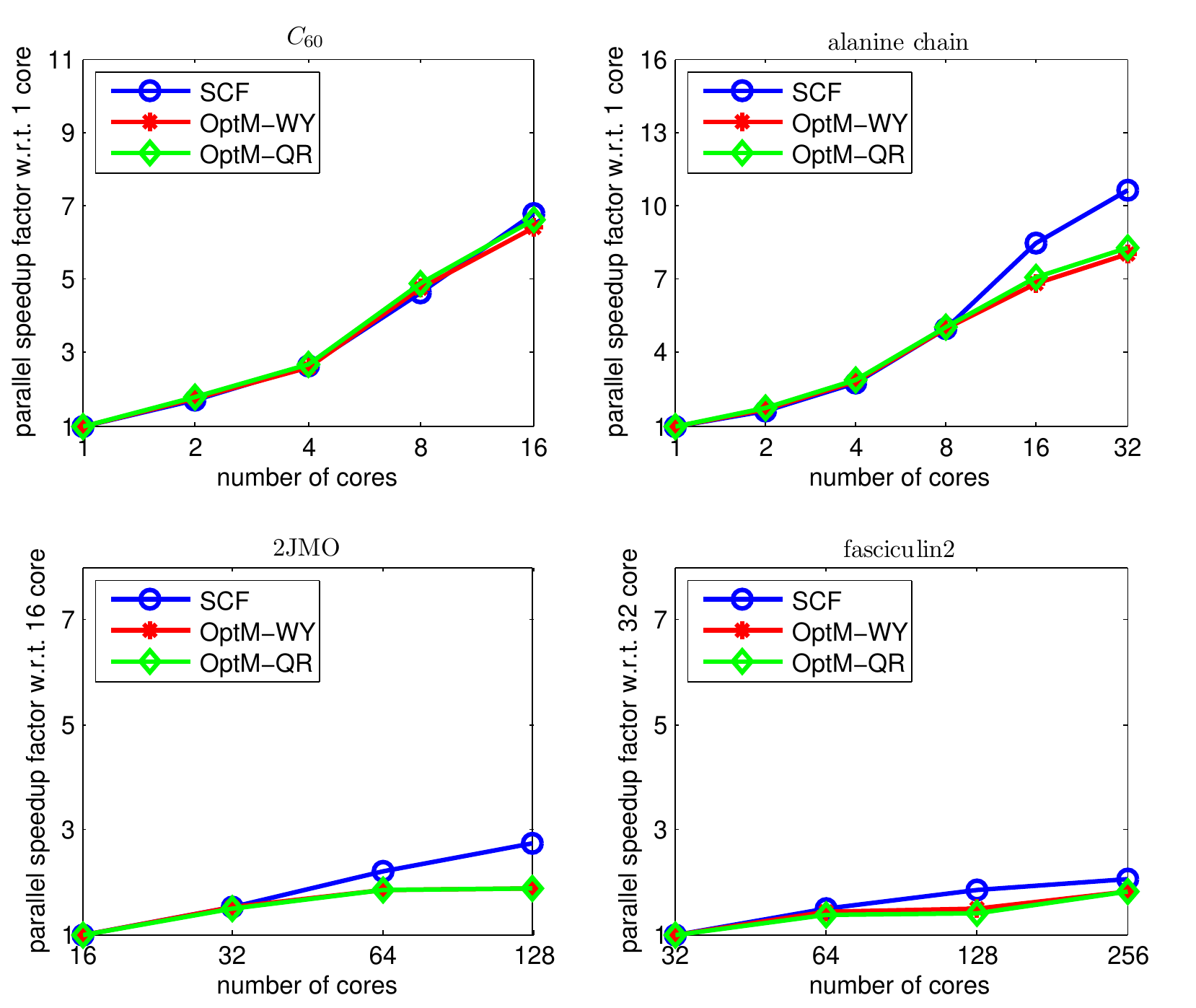}
\end{center}
\caption{The speedup factor with respect to T0}\label{figure:T0speedup}
\end{figure}

\begin{figure}[htb]
\begin{center}
\includegraphics[width=0.9\textwidth,height=0.7\textwidth ]{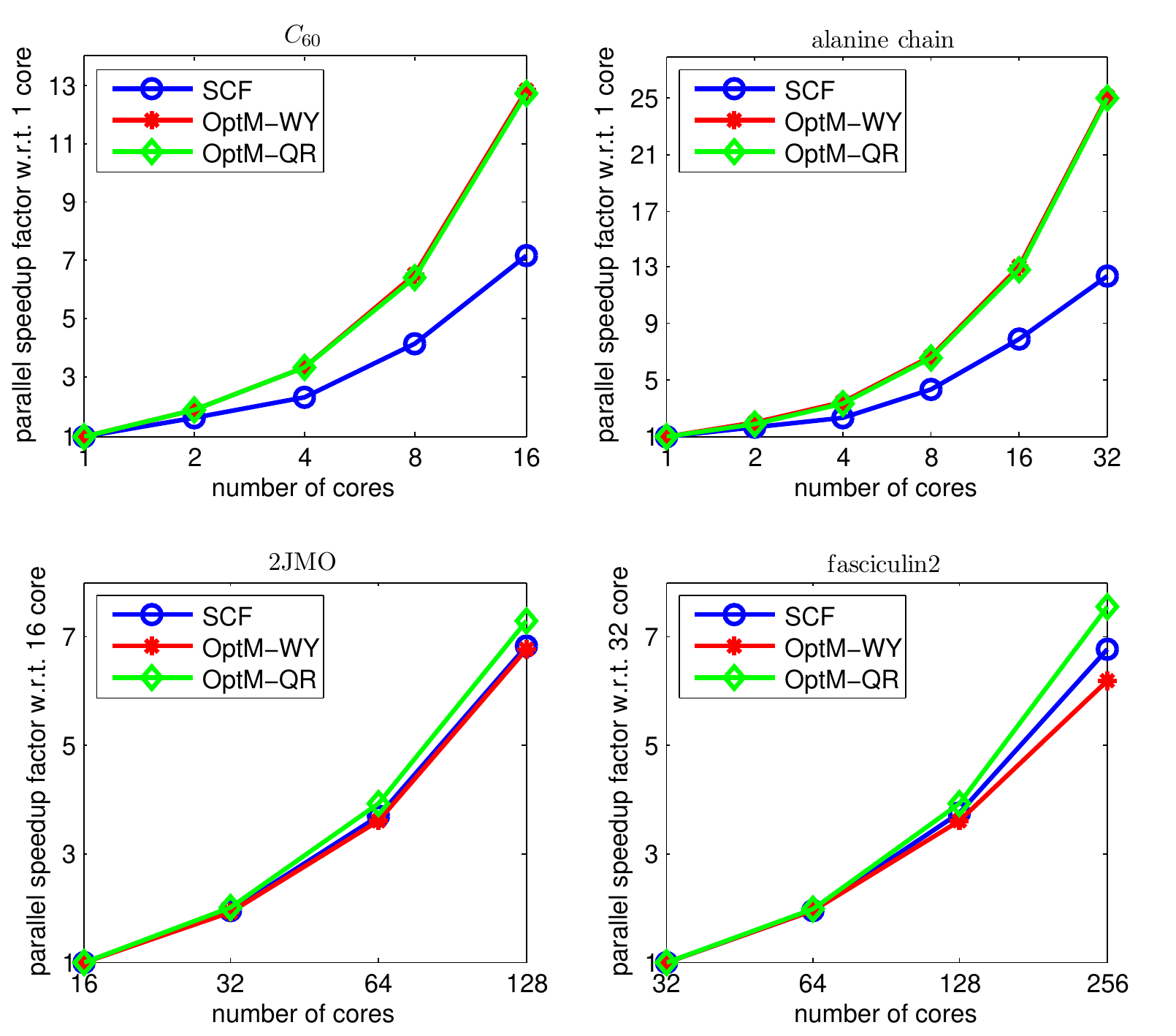}
\end{center}
\caption{The speedup factor with respect to T1}\label{figure:T1speedup}
\end{figure}

Figures \ref{figure:T0speedup} and \ref{figure:T1speedup} show the speedup
factors of T0 and T1, respectively. We can see that the scalability of T0 is not
good on 2JMO and fasciculin2. The performance of SCF is not always the same as
the gradient type methods because that their time and proportion of calculating the total energy,
gradients and Hamiltonian are different.
 On the other hand, OptM-QR is better than SCF in terms of T1. OptM-WY behaves
 similar to OptM-QR on $C_{60}$ and alanine chain, but it is worse on 2JMO and
 fasciculin2. The reason is that the complexity at each iteration also depends
 on the number of columns $p$. The Cholesky factorization of a $p\times p $
 matrix in OptM-QR costs $\frac{1}{3}p^3$ while the
calculation of $\left( I + \frac{\tau}{2} V^\top  U \right)
^{-1} V^\top X$ in OptM-WY needs $\frac{40}{3}p^3+\mathcal{O}(p^2)$. These two
operations are not parallelized in our current implementation.
Consequently, OptM-QR has a slightly higher parallel speedup factor than OptM-WY.
 We next present the ratio $T_0/(T_0+T_1)$ in  Figure \ref{figure:T0ratio}.  It shows that the most time
 consuming part of SCF is T1, in particular, for larger systems, due to the
 eigenvalue computation. On the other hand, T0 accounts for a large proportion
 of OptM-QR and OptM-WY. Hence, the scalability of our gradient type methods can be further
 improved if the efficiency of $T0$ can be enhanced.

\begin{figure}[htb]
\begin{center}
\includegraphics[width=0.9\textwidth,height=0.7\textwidth]{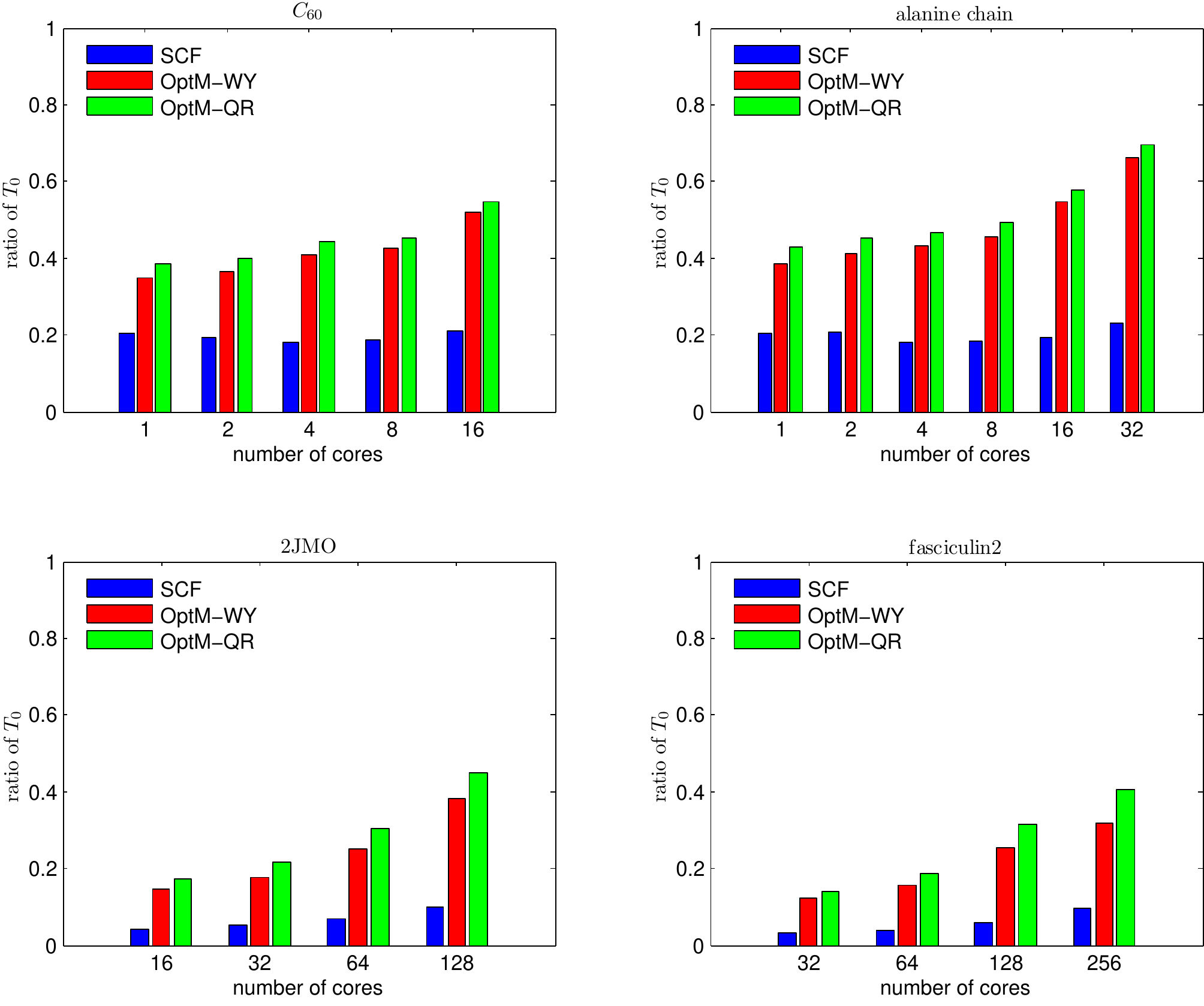}
\end{center}
\caption{The ratio of $T_0/(T_0+T_1)$.}\label{figure:T0ratio}
\end{figure}

Finally, we investigate the sensitivity of the gradient type methods with
respect to the number of cores on the systems $C_{60}$ and alanine chain. The parameters
of each run are the same. The results are presented in Table \ref{table:para}.
We can see that the total number of iterations of the gradient type methods are
not always the same. The reason may be that the BB step size is sensitive to
 numerical errors. Hence, a more robust method for choosing the step size is
 expected.

\begin{table}[h]
  \begin{center}
    \tabcolsep=4pt
    \begin{tabular}{|c|ccc|ccc|ccc|}
    \hline
           &  \multicolumn{3}{c|}{OptM-WY} &  \multicolumn{3}{c|}{OptM-QR} &
           \multicolumn{3}{c|}{SCF} \\
    \cline{2-10}
    cores  &  resi & Iter & cpu(s) &  resi & Iter & cpu(s) &  resi & Iter & cpu(s) \\
    \hline
    \multicolumn{10}{|c|}{$C_{60}\quad p=120\quad n=191805$\quad{\bf $\varepsilon_g\,=\,1e-6$}}\\
    \hline
     1 & 9.96e-07 &	226	& 904 &	7.89e-07 & 245 & 891 & 7.05e-07 & 23 & 1649\\
     2 & 6.79e-07 &	260	& 562 &	8.79e-07 & 244 & 484 & 6.67e-07 & 23 & 1123\\
     4 & 6.37e-07 &	224	& 294 &	9.97e-07 & 241 & 291 & 6.67e-07 & 23 &  672\\
     8 & 3.72e-07 &	237	& 165 &	9.51e-07 & 236 & 153 & 6.52e-07 & 23 &  376\\
    16 & 5.68e-07 &	239	& 100 &	9.53e-07 & 242 &  96 & 6.51e-07 & 23 &  225\\
    \hline
    \multicolumn{10}{|c|}{alanine chain$\quad p=132\quad n=293725$\quad{\bf $\varepsilon_g\,=\,1e-6$}}\\
    \hline
     1 & 9.93e-07 & 1605 &	11984 & 9.95e-07& 1591 & 10671 & --  & -- & --\\
     2 & 9.75e-07 &	2086 &	 8534 & 7.59e-07& 1949 &  7475 & --  & -- & --\\
     4 & 9.50e-07 & 1856 &	 4344 & 9.75e-07& 1584 &  3442 & --  & -- & --\\
     8 & 9.06e-07 &	1985 &	 2526 & 9.97e-07& 1502 &  1760 & --  & -- & --\\
    16 & 9.98e-07 & 1816 &	 1410 & 8.72e-07& 1828 &  1318 & --  & -- & --\\
    32 & 9.64e-07 &	2082 &	 1123 & 5.57e-07& 1413 &   727 & --  & -- & --\\
    \hline
    \end{tabular}
  \end{center}
  \caption{Numerical results with respect to the number of cores. The results of SCF are
  not reported for alanine chain because SCF failed as shown in Table \ref{table:LCAOh3}
.}\label{table:para}
\end{table}


\subsection{Numerical results for the OFDFT model}

Our numerical analysis for OFDFT is based on aluminum crystal,
where \eqref{eq:TFW} is used as KEDF and external
potential is the GNH (Goodwin-Needs-Heine) pseudopotential
 \cite{goodwin-needs-heine90}:
\begin{align}
V_{ext}(r)=\frac{2}{\pi}\int_0^{\infty}\frac{\sin(rt)}{rt}\Big(
(Z-AR)\cos(Rt)+A\frac{\sin(Rt)}{t}\Big)e^{-(\frac{t}{R_c})^6}dt,
\end{align}
where $Z$ is the number of valence electrons, 
 $R=1.15$,
$R_c=3.5$ and $A=0.1107$. Several finite systems with fixed atomic positions are simulated.

We compare the performance of SCF, OptM-WY
and OptM-QR. They are further embedded in the adaptive mesh refinement method
Algorithm \ref{alg:AdpConOptM}. Since OFDFT is not available in
Octopus, we implement all methods in the package RealSPACES, 
which is developed based on the
platform PHG (Parallel Hierarchical Grid) \cite{phg}.
The initial guess is generated by the pseudo-wave functions of
aluminum and the initial mesh is produced by RealSPACES. A lattice spacing of
7.559 a.u. is used for the size of the unit cell.
 The maximum number of iterations for SCF and
OptM-QR and OptM-WY is 25 and 200, respectively.  We terminated all methods if
$\|\nabla E(X_k)\|_F\le \varepsilon_g = 10^{-5}$.

A summary of numerical
results is reported in Table \ref{table:compare-OF}, where $size$, $N_{Al}$,
$E^0_p$, and $E_b$ stand for the number of unit cells, the total number of
aluminum atoms, the ground state energy per atom, and the binding energy.
The binding energy is evaluated by
$ E_b = \frac{E_0-N_{Al} E_s}{N_{Al}}$, where $E_0$ is the ground state
energy calculated by \eqref{eq:OFeng}, and $E_s=-52.800704 (eV)$ is the
energy for single aluminum atom. $E^0_p$, $E_b$ and $E_0$ are all measured in
eV. The number $n$ denotes the total number of degrees of freedom in the
final adaptive step.

\begin{table}
\centering
\begin{tabular}{|c|cccc|}\hline
    solver    & { $E^0_p$(eV)} &  $E_b$(eV)  & {$n$}     & {cpu(s)}\\
    \hline
    \multicolumn{5}{|c|}{$size=4\times 4\times 4$\quad$N_{Al}=365$
    \quad{$cores=32$}}\\\hline
      SCF      &  -57.036037    &  -4.235333  &  835908    & 1557  \\
    OptM-WY    &  -57.036037    &  -4.235333  &  835904    & 867   \\
    OptM-QR    &  -57.036038    &  -4.235334  &  835895    & 756   \\\hline
    \multicolumn{5}{|c|}{$size=7\times 7\times 7$\quad$N_{Al}=1688$
    \quad{$cores=64$}}\\\hline
      SCF      &  -57.150302    &  -4.349598  &  4486542   & 8368   \\
    OptM-WY    &  -57.150302    &  -4.349598  &  4485928   & 5245   \\
    OptM-QR    &  -57.150302    &  -4.349598  &  4485919   & 4732   \\\hline
    \multicolumn{5}{|c|}{$size=10\times 10\times 10\quad N_{Al}=4631$
    \quad{\bf $cores=128$}}\\\hline
      SCF      &  -57.628512    &  -4.827808  &  13411386  & 15588   \\
    OptM-WY    &  -57.628513    &  -4.827809  &  13411388  & 9412    \\
    OptM-QR    &  -57.628513    &  -4.827809  &  13411373  & 8852    \\\hline
    \multicolumn{5}{|c|}{$size=12\times 12\times 12$\quad$N_{Al}=7813$
    \quad{$cores=128$}}\\\hline
      SCF      &  -58.093705    &  -5.293001  &  45010875  & 58678    \\
    OptM-WY    &  -58.093706    &  -5.293003  &  45010864  & 30645   \\
    OptM-QR    &  -58.093707    &  -5.293003  &  45010826  & 26457   \\\hline
\end{tabular}
\caption{Numerical results computed by the adaptive mesh refinement method.}\label{table:compare-OF}
\end{table}

Table \ref{table:compare-OF} shows that the ground state energy per
atom converges as the size of the system is increased.  The gradient type
methods are more efficient than SCF, and
OptM-QR  is slightly  better than OptM-WY. We
should point out there exists difference on the wave
functions on the adaptive grids obtained from different gradient type methods.
Hence, the final total
number of degrees of freedom $n$ may be different even if their ``size'' are the same.
\if false
\comm{The results computed by other methods in the literature are presented
in Table \ref{table:energy-per}. We can see that our
numerical results are close to the ones in Table \ref{table:energy-per}.}


\begin{table}[!ht]
\centering
\begin{tabular}{|c|c|}\hline
Method     &   $E^0_p$(eV)    \\\hline
KSDFT/BLPS \cite{shin-ram-huang-hung-cater09} &   -57.9550     \\
KSDFT/NLPP \cite{carling-carter03}            &   -57.1940     \\
KSDFT/LPP \cite{carling-carter03}             &   -58.3321     \\
OFDFT/SM \cite{smargiassi-madden}             &   -58.4403     \\
OFDFT/WGC \cite{wang-govind-carter98}         &   -58.3340     \\
OFDFT/WT \cite{wang-teter92}                  &   -58.3303     \\\hline
\end{tabular}
\caption{The ground state energy for single aluminum obtained by different methods,
where BLPS, NLPP, LPP are the names of the pseudopotential, and SM, WGC, WT are
the KEDF
formulas.}\label{table:energy-per}
\end{table}
\fi

The change of the residuals versus the iteration history on a
particular mesh is presented in Figure \ref{fig:AlOptM4} for systems $Al_{1688}$
and $Al_{4631}$, respectively, where $Al_{1688}$ denotes the
aluminum $7\times 7\times 7$ cluster with 1688 aluminum atoms and similar
notations are used for $Al_{365}$, $Al_{4631}$ and $Al_{7813}$. The gradient type methods are
able to reduce the residuals steadily although they may be increased at some
iterations. The contours of the ground state charge density and their corresponding
adaptive mesh distributions are shown in Figure \ref{fig:DenDistribution} for an
intuitive illustration of our approaches.


\begin{figure}[h]
\begin{center}
\includegraphics[width=0.8\textwidth,height=0.6\textwidth]{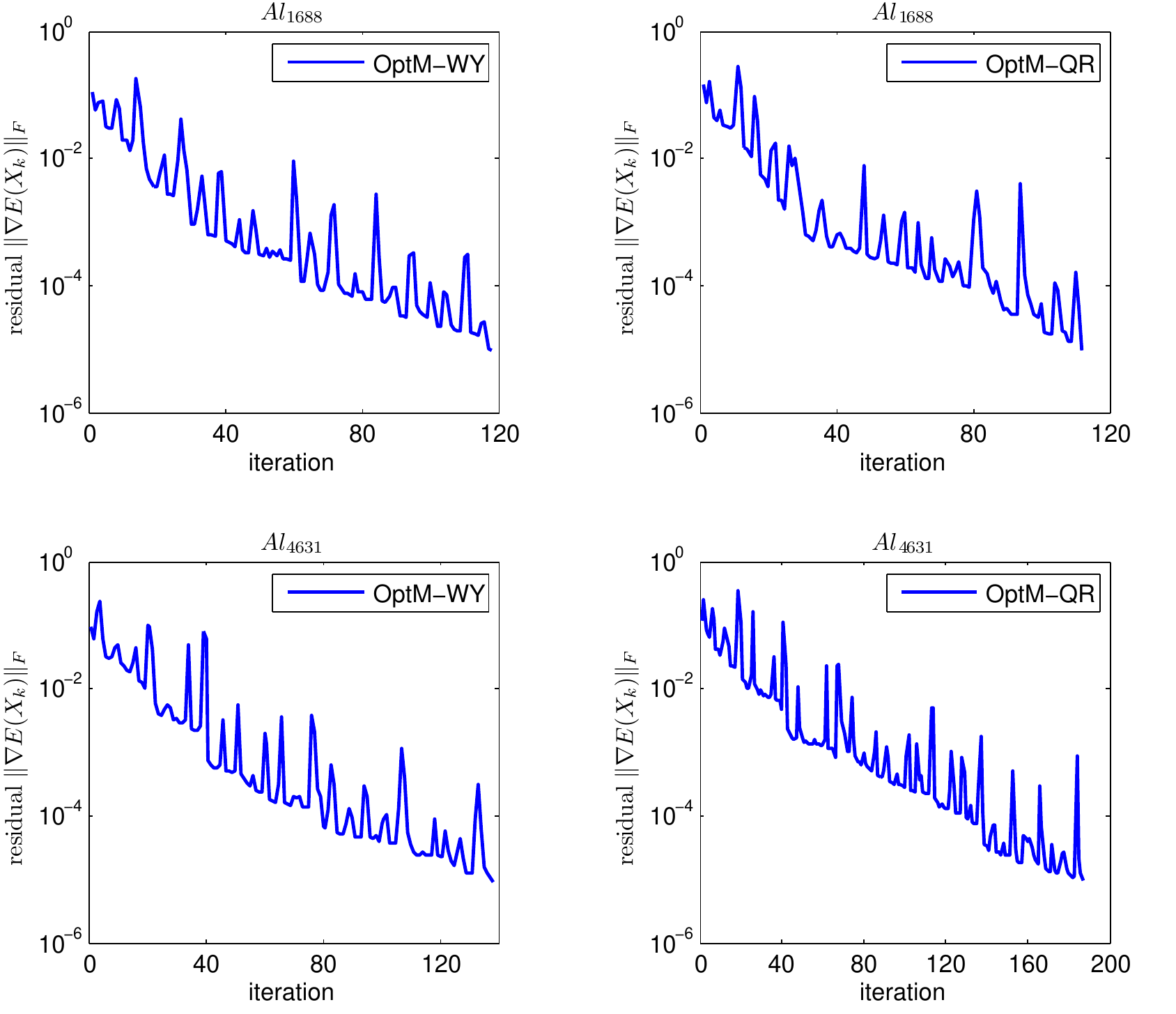}
\end{center}
\caption{Residuals $\|G-XG^TX\|_{F}$. Top: $Al_{1688}$; Bottom: $Al_{4631}$.}\label{fig:AlOptM4}
\end{figure}

\begin{figure}[!ht]
\centering
\subfigure[] {\label{fig:al1688-0}
\includegraphics[width=4.75cm]{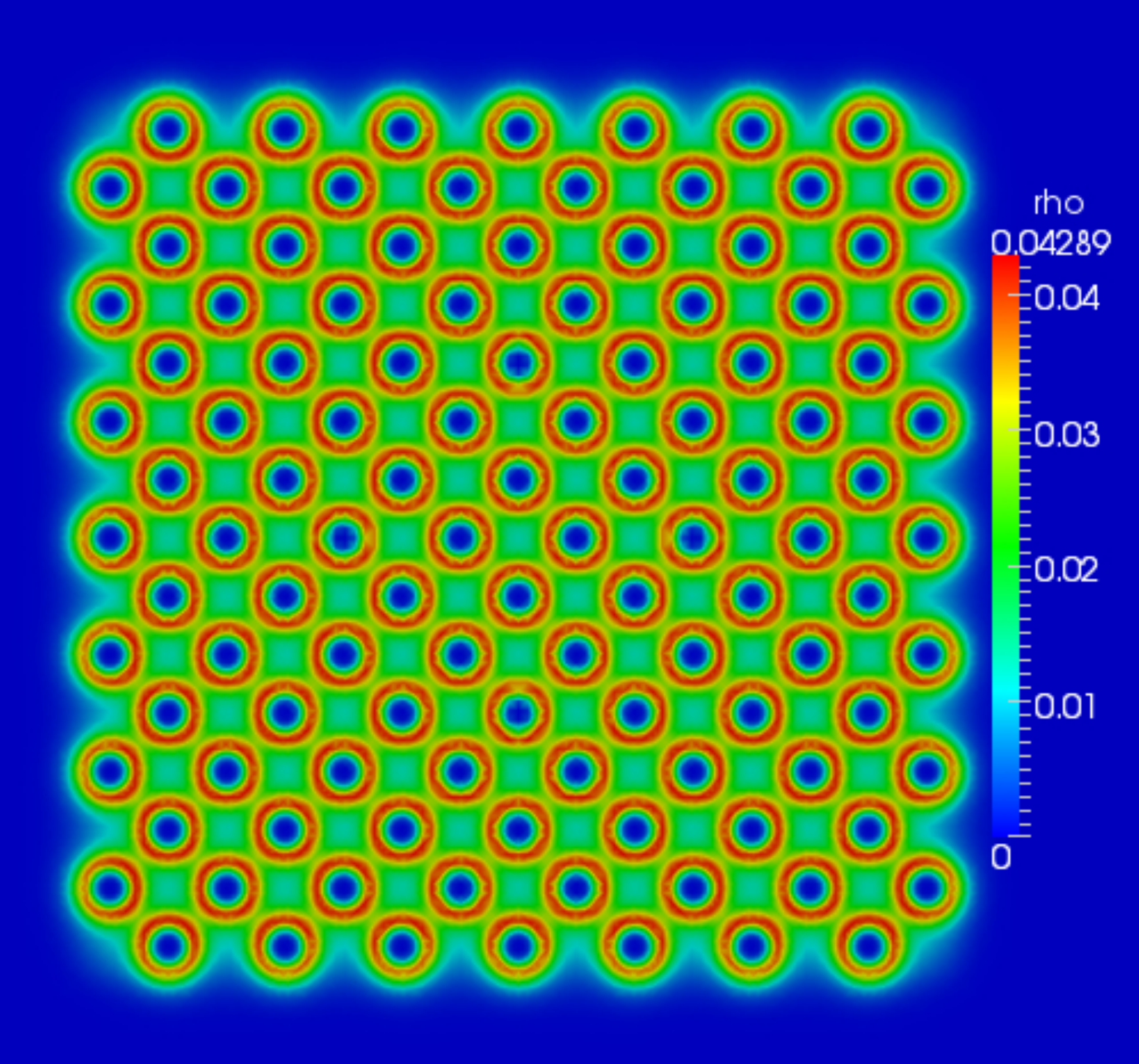}} \hskip 0.8cm
\subfigure[] {\label{fig:al1688-3.8}
\includegraphics[width=4.4cm]{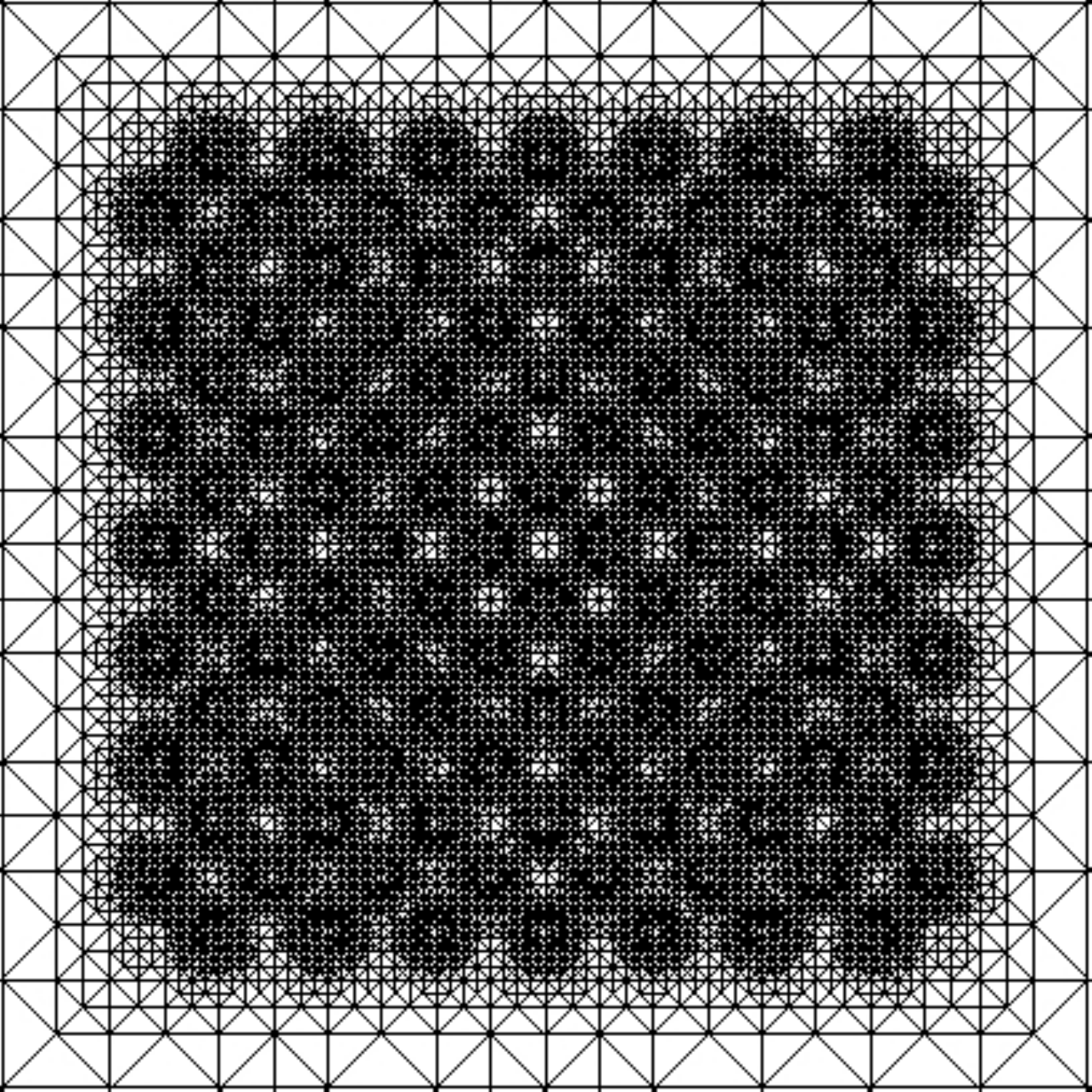}}
\subfigure[] {\label{fig:al4631-0}
\includegraphics[width=4.75cm]{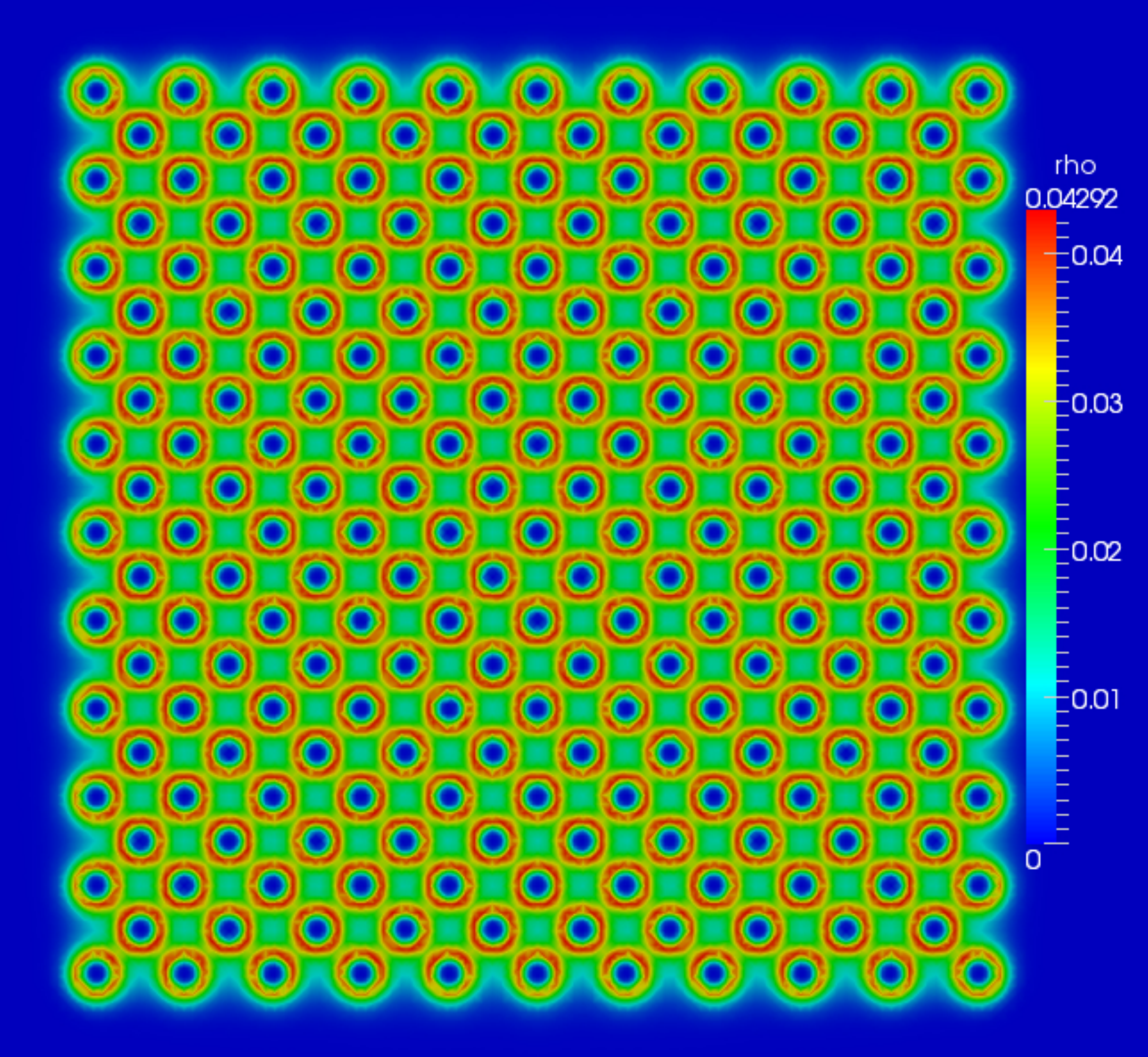}} \hskip 0.8cm
\subfigure[] {\label{fig:al4631-3.8}
\includegraphics[width=4.4cm]{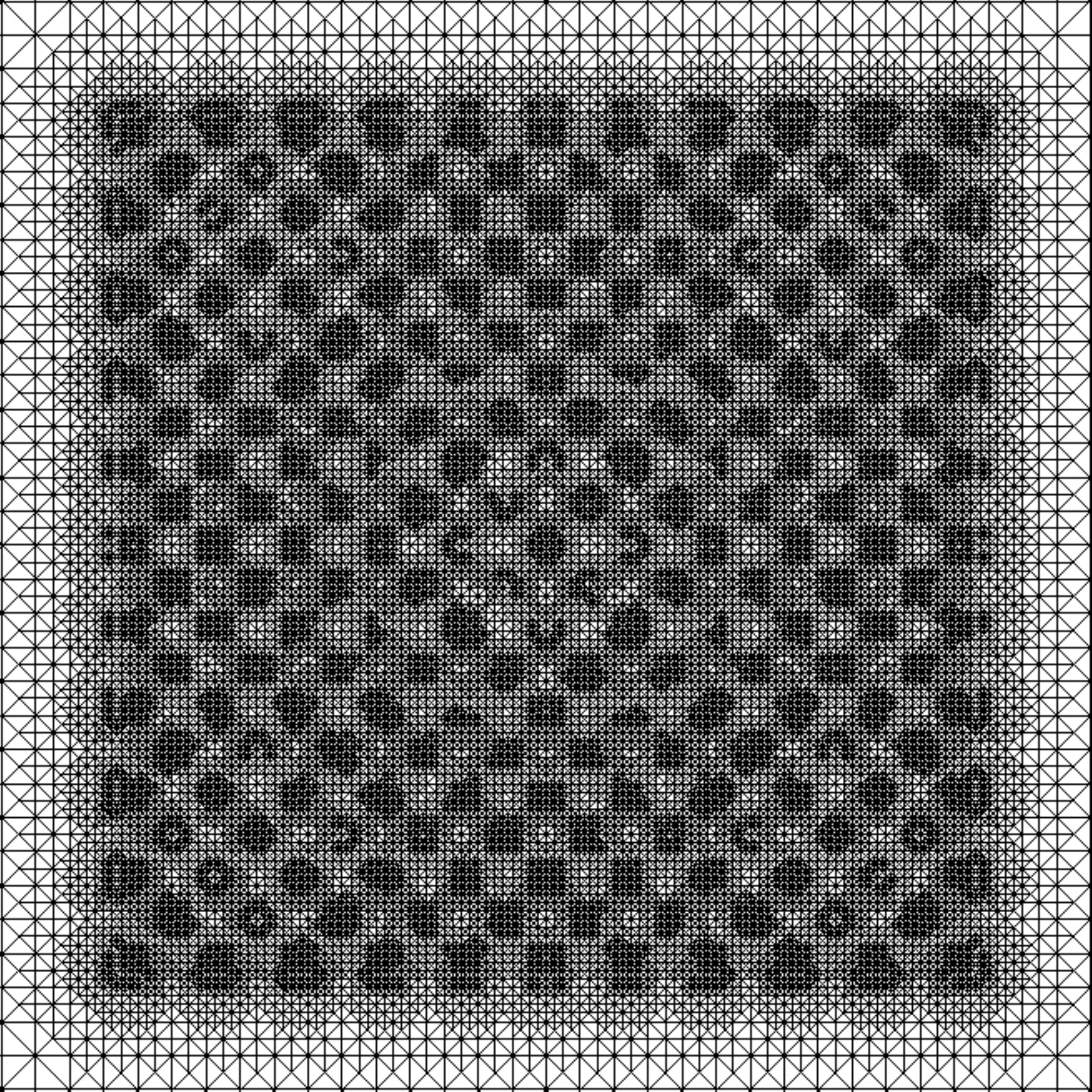}}
\caption{(a) and (c) are the contours of the ground state charge density at
plane $z=0$ for $Al_{1688}$ and $Al_{4631}$, respectively. (b) and (d) are the
corresponding
adaptive mesh distribution of (a) and (c), respectively.} \label{fig:DenDistribution}
\end{figure}



We next examine the parallel scalability of the gradient type methods.
Similar to  KSDFT, the wall clock time
is split into the T0 and T1 parts,
 where T0 includes the wall clock time on computing the total energy, gradients and
 Hamiltonian, while all other wall clock time is counted as T1. The speedup
 factors of T0 and T1 defined in \eqref{eq:speedup-factor} for systems $Al_{365}$ and
 $Al_{1688}$ are presented in  Figures \ref{fig:AlspeedupT0} and
 \ref{fig:AlspeedupT1}, respectively. We can see that the difference between
 OptM-WY and OptM-QR is small. The reason is that the parallel
 scalability of each method only depends on the spatial
degrees of freedom $n$ in the case of $p=1$. The difference between T0 and T1 is
 also small.  Consequently, the overall  parallel scalability is high. The
 performance of OptM-WY and OptM-QR is at least comparable to that of SCF.

%

\begin{figure}[h]
\begin{center}
\includegraphics[width=10cm]{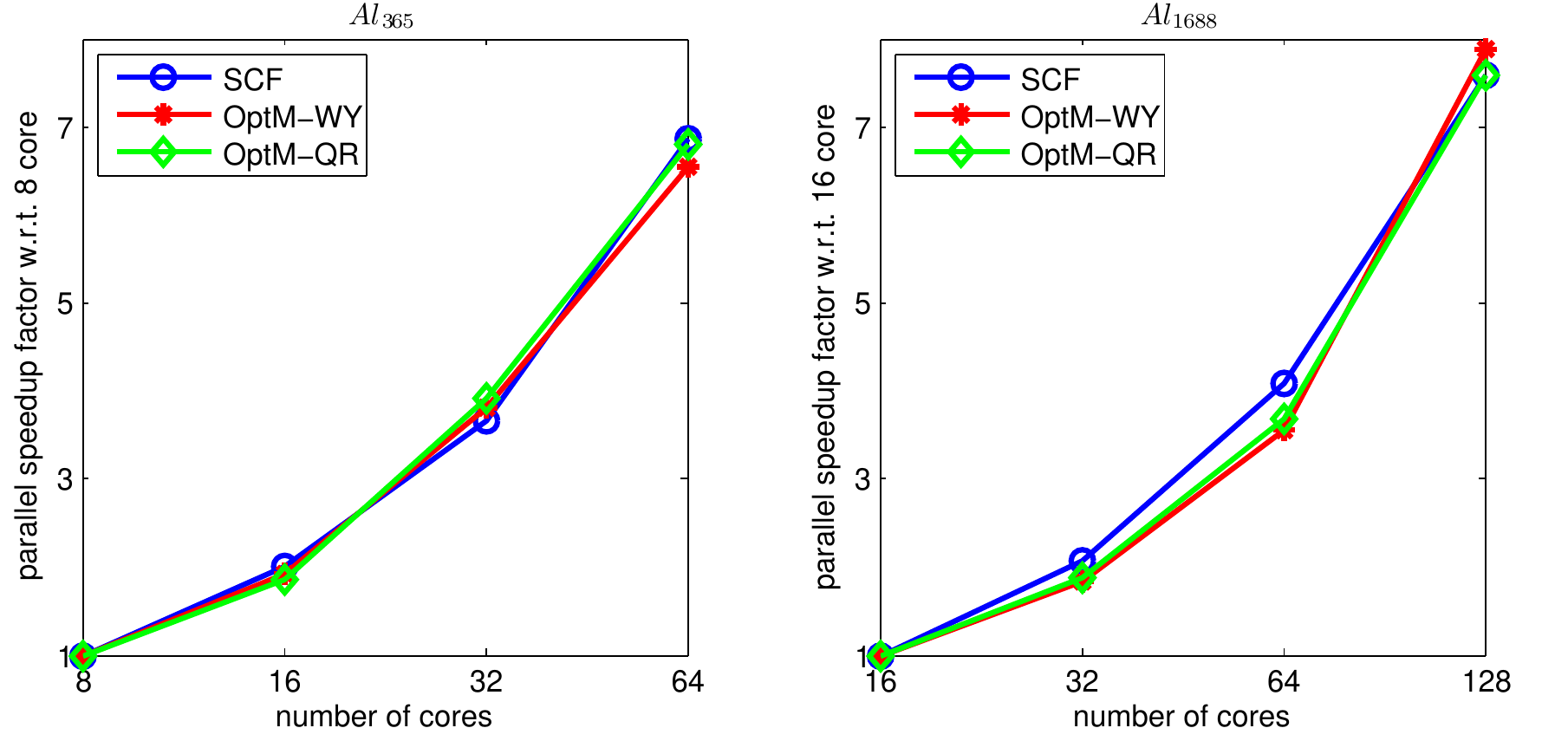}
\end{center}
\caption{The speedup factors for aluminum crystal with respect to T0.}\label{fig:AlspeedupT0}
\end{figure}

\begin{figure}[h]
\begin{center}
\includegraphics[width=10cm]{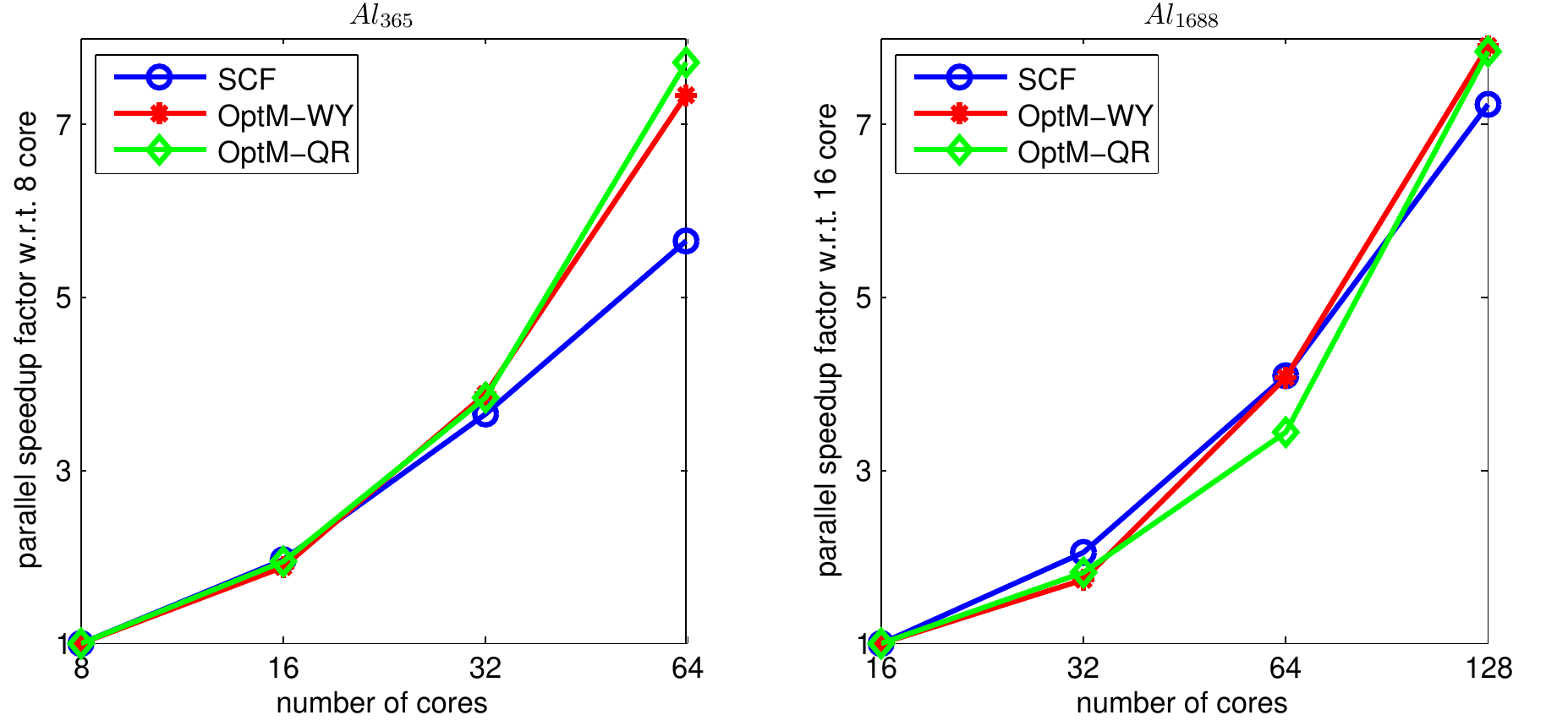}
\end{center}
\caption{The speedup factors for aluminum crystal with respect to T1.}\label{fig:AlspeedupT1}
\end{figure}

Finally, the results obtained by refining the mesh uniformly are presented in
Table \ref{table:uniform-OF}. This strategy uses the same number of meshes as
the adaptive mesh refinement method in Table \ref{table:compare-OF}, but
these meshes are refined uniformly from their coarser levels. We can conclude
from these two tables that the adaptive mesh refinement strategy
can greatly reduce the total number of degrees of freedom as well as the
computational cost.


\begin{table}
\centering
\begin{tabular}{|c|cccc|}\hline
    {solver}       & { $E^0_p$(eV)} &  $E_b$(eV)  &   {$n$}   & {cpu(s)}\\
    \hline
    \multicolumn{5}{|c|}{$size=4\times 4\times 4$\quad$N_{Al}=365$
    \quad{$cores=32$}}\\\hline
    OptM-WY       &  -57.092410     &   -4.291706 & 1432850   & 1448  \\
    OptM-QR       &  -57.092441     &   -4.291737 & 1432850   & 1288  \\\hline
    \multicolumn{5}{|c|}{$size=7\times 7\times 7$\quad$N_{Al}=1688$
    \quad{$cores=64$}}\\\hline
    OptM-WY       &  -57.151769     &   -4.351065 & 16974593  & 20688  \\
    OptM-QR       &  -57.151767     &   -4.351063 & 16974593  & 17356  \\\hline
    \multicolumn{5}{|c|}{$size=10\times 10\times 10\quad N_{Al}=4631$
    \quad{\bf $cores=128$}}\\\hline
    OptM-WY       &  -57.784717     &   -4.984013 & 37991437  & 19347  \\
    OptM-QR       &  -57.784860     &   -4.984156 & 37991437  & 18362  \\\hline
\end{tabular}
\caption{Numerical results computed by the uniformly mesh refinement method}\label{table:uniform-OF}
\end{table}

\section{Concluding remarks} \label{sec:conclusion}
In this paper, we study gradient type methods for solving the KSDFT and
OFDFT models in electronic structure calculations.
Unlike the commonly used SCF iteration, these approaches
do not rely on solving linear eigenvalue problems.
 The main components of our approaches are gradients on the Stiefel manifold and operations for preserving the
orthogonality constraints. They are cheaper and often  have better parallel
scalability than eigenvalue computation.
A specific form uses the QR factorization to orthogonalize  the gradient
step on the Stiefel manifold explicitly.  To the best of the authors' knowledge,  it is the first time that the QR-based
method is systematically studied for electronic
structure calculations.
The
gradient methods is further improved by the state-of-the-art
acceleration techniques such as Barzilai-Borwein steps and non-monotone
line search with global convergence guarantees.
We implement our methods based on
two real space software packages, i.e., Octopus for KSDFT and RealSPACES for OFDFT,
respectively. Numerical experiments show that our methods can be more efficient and robust than SCF on many
instances. Their parallel efficiency is ideal as long as the evaluation of
the total energy functional and their gradients is efficient.
We believe that
these methods are powerful techniques in simulating large and complex
systems.

\section*{Acknowledgements}
X. Zhang, J. Zhu and A. Zhou would like to thank Dr. H. Chen,
Prof. X. Dai, Dr. J. Fang, Dr. X. Gao, Prof. X. Gong and Dr. Z. Yang for their
stimulating discussions and cooperations on electronic structure calculations.
 Z. Wen would like to thank Humboldt
 Foundation for the generous support, Prof. Michael Ulbrich for
hosting his visit at Technische Universit\"at
M\"unchen and Dr. Chao Yang for discussion on Octopus.

\bibliographystyle{siam}
\bibliography{bibDFT,optimization}

\end{document}